\newif\iffull\fulltrue
\iffull
\documentclass[11pt]{article}
\usepackage[margin=1in]{geometry}
\else
\documentclass[USenglish,cleveref,autoref,thm-restate]{socg-lipics-v2021}
\fi

\usepackage[T1]{fontenc}
\usepackage[utf8]{inputenc}

\usepackage{amsfonts,amsthm,amsmath,amssymb}
\usepackage{cite,microtype,graphicx,hyperref}

\iffull
\usepackage[capitalize,nameinlink]{cleveref}

\newtheorem{theorem}{Theorem}
\newtheorem{lemma}[theorem]{Lemma}
\newtheorem{corollary}[theorem]{Corollary}
\newtheorem{observation}[theorem]{Observation}
\fi

\theoremstyle{definition}

\iffull
\newtheorem{definition}[theorem]{Definition}
\fi

\newtheorem{algorithm}[theorem]{Algorithm}
\newtheorem{subroutine}[theorem]{Subroutine}

\newcommand{\smallcap}[1]{\textsc{#1}\index{#1@$\textsc{#1}$}}
\newcommand{\offline}{\smallcap{offline}}
\newcommand{\inhull}{\smallcap{inhull}}
\newcommand{\nham}{\smallcap{\#ham}}
\newcommand{\npath}{\smallcap{\#path}}
\newcommand{\nsurround}{\smallcap{\#surround}}
\newcommand{\npoly}{\smallcap{\#poly}}

\title{Non-crossing Hamiltonian Paths and Cycles\\ in Output-Polynomial Time}

\iffull
\author{David Eppstein\thanks{Department of Computer Science, University of California, Irvine. Research supported in part by NSF grant CCF-2212129.}}
\date{ }

\else
\titlerunning{Non-crossing Hamiltonian Paths and Cycles in Output-Polynomial Time}
\author{David Eppstein}{Computer Science Department, University of California, Irvine, USA}{eppstein@uci.edu}{}{Research supported in part by NSF grant CCF-2212129}

\authorrunning{D. Eppstein}
\Copyright{David Eppstein}

\fi

\begin{document}

\maketitle

\begin{abstract}
We show that, for planar point sets, the number of non-crossing Hamiltonian paths is polynomially bounded in the number of non-crossing paths, and the number of non-crossing Hamiltonian cycles (polygonalizations) is polynomially bounded in the number of surrounding cycles. As a consequence, we can list the non-crossing Hamiltonian paths or the polygonalizations, in time polynomial in the output size, by filtering the output of simple backtracking algorithms for non-crossing paths or surrounding cycles respectively. To prove these results we relate the numbers of non-crossing structures to two easily-computed parameters of the point set: the minimum number of points whose removal results in a collinear set, and the number of points interior to the convex hull. These relations also lead to polynomial-time approximation algorithms for the numbers of structures of all four types, accurate to within a constant factor of the logarithm of these numbers.
\end{abstract}

\section{Introduction}

In how many ways can we ``connect the dots'', turning planar points into the vertices of a simple polygon? Despite heavy study, the answer is still unclear. Steinhaus proved in the 1960s that general-position points always have at least one polygonalization~\cite{Ste-64}, but the condition is too strong: non-collinearity suffices. Points in convex position have only one polygonalization; other inputs can have exponentially many, but with upper and lower bounds that are far from matching~\cite{GarNoyTej-CGTA-00,ShaSheWel-JCTA-13}. The complexity of counting polygonalizations is unknown~\cite{MitORo-IJCGA-01,MarMil-SoCG-16,Epp-DCG-20}. Although all polygonalizations can be listed in singly-exponential time~\cite{Wet-JoCG-17,YamAviHor-DAM-21}, it was unknown (prior to our work) how to list them more quickly when there are few, for instance in polynomial time per output or in time polynomial in the output size.

Our main result is that both polygonalizations and a closely related structure, non-crossing Hamiltonian paths, can be listed in time polynomial in the output size by simple backtracking algorithms.  These algorithms search spaces of easier-to-list structures, the \emph{surrounding polygons}~\cite{YamAviHor-DAM-21} and non-crossing paths, respectively. To prove these new results,
following our work on monotone parameters of point sets~\cite{Epp-18},
we relate the numbers of all four types of structures (polygonalizations, surrounding polygons, non-crossing Hamiltonian paths, and non-crossing paths) to two easily-computed parameters that depend only on the order-type of the points: the smallest number of points whose removal results in a collinear subset, and the number of points interior to the convex hull. These relations imply that the number of polygonalizations is at most polynomial in the number of surrounding polygons, and that the number of non-crossing Hamiltonian paths is at most polynomial in the number of non-crossing paths. Therefore, an algorithm for the easier-to-list structures will take output-polynomial time when its output is filtered to generate only the harder-to-list structures. Our methods also provide a polynomial-time approximation algorithm for counting these structures, obtaining a constant approximation ratio with respect to the logarithm of the count.

We do not calculate explicitly the exponents of the polynomials relating these numbers of structures, as our upper bounds are quite imprecise. Instead, in a final section, we describe point sets whose exponential numbers of non-crossing structures can be calculated precisely. These examples demonstrate that the exponent of paths in terms of Hamiltonian paths can be at least $\log_23\approx 1.585$ and that the exponent of surrounding polygons in terms of polygonalizations can be at least $\log_2\bigl((3+\sqrt5)/2\bigr)\approx 1.388$. 

Output-polynomial time bounds, such as we prove, are not as good as a time bound that multiplies the output size by a polynomial of the input size, and even less good than polynomial delay for each output structure. Nevertheless, our results represent a significant improvement on previously known time bounds, which are singly exponential even for inputs whose output size is subexponential. 

\subsection{Related work}

A \emph{planar straight line graph}, with given points as vertices, consists of line segments having the points as endpoints, with no line segment passing through a given point and no two line segments intersecting except at a shared endpoint. When the form of the resulting graph is constrained, one obtains non-crossing structures of various types, including triangulations (non-crossing maximal planar graphs), non-crossing spanning trees, and polygonalizations (non-crossing Hamiltonian cycles). Extensive research in discrete and computational geometry has sought upper and lower bounds for numbers of non-crossing structures, and studied algorithmic problems of counting or listing these structures for a given point set, or of finding a non-crossing structure that is optimal for some objective function~\cite{AjtChvNew-TPC-82,AloRajSur-FI-95,FlaNoy-DM-99,GarNoyTej-CGTA-00,DumTot-DCG-10,ShaShe-CPC-13,HofSchSha-30E-13,RazWel-RCS-11,ShaSheWel-JCTA-13,AlvBriCur-DCG-15,MarMil-SoCG-16,CheChoEst-DAM-22}. Many problems of this sort have algorithms for listing all graphs, with polynomial time per output, based on systems of local moves that link the state space into a connected structure~\cite{AviFuk-DAM-96,Bes-CGTA-02,AicAurHue-GC-07,KatTan-DCG-09,WuChaPai-TCS-11,Tan-EoA-16,Wet-JoCG-17,YamAviHor-DAM-21}.

However, for polygonalizations and  non-crossing Hamiltonian paths no such structure is known, and natural systems of local moves that change two or three edges at a time are known not to connect all polygonalizations~\cite{HerHouHur-TCS-02}. Certain generalizations of polygonalizations have state spaces connected by local moves~\cite{DamFlaORo-TCS-10,YamAviHor-DAM-21}, but this does not directly yield fast algorithms for polygonalizations, because of the many non-polygonalizations in these state spaces. As a 2011 survey of Welzl summarizes, ``Basically nothing is known for related algorithmic questions (determining the number of simple polygonizations for a given point set, enumerating all simple polygonizations)''~\cite{Wel-CCCG-11}.
The shortest polygonalization is the $\mathsf{NP}$-hard Euclidean traveling salesperson tour~\cite{QuiSup-AMM-65}, and several other optimal polygonalizations are also $\mathsf{NP}$-hard~\cite{Fek-DCG-00,DemFekKel-JEA-22}; the complexity of the longest polygonalization is another unknown~\cite{DumTot-DCG-10}.

Past work on  non-crossing Hamiltonian paths includes approximation to the longest path~\cite{AloRajSur-FI-95,DumTot-DCG-10} and the existence of properly colored paths for colored point sets~\cite{BanBanBho-TCS-21}. The number of non-crossing Hamiltonian paths can range from one, for collinear points, to exponentially large; for instance, $n\ge 2$ points in convex position have exactly $n2^{n-3}$ paths~\cite{A001792}.

The surrounding polygons that we use for our polygonalization algorithm are another class of planar straight line graphs: they are the simple polygons that use a subset of the input points as vertices, and contain all of them. Yamanaka et~al.~\cite{YamAviHor-DAM-21} introduced these polygons, and showed how to list them in polynomial time per polygon, from which they derived a singly-exponential time bound and polynomial space bound for listing polygonalizations. Despite this progress they were unable to obtain an output-sensitive time bound for polygonalizations. It is their algorithm that we follow here, with a new output-sensitive analysis.  Yamanaka et~al. also prove that, in the worst-case exponential bounds for numbers of polygonalizations and surrounding polygons, the bases of the exponentials are within 1 of each other; however, this analysis does not imply our stronger result, that on arbitrary instances the numbers of these two types of polygons are bounded by polynomials of each other. The examples from our final section confirm theoretically the empirical results of Yamanaka et~al. that polygonalizations and surrounding polygons can grow at different rates.

\section{Two simple backtracking algorithms}

In this section we outline two simple algorithms (one a standard backtracking search, and the other the reverse search algorithm of Yamanaka et~al.~\cite{YamAviHor-DAM-21}) for listing non-crossing paths and surrounding cycles. These known algorithms will be the ones we use to prove our new time bounds for listing non-crossing Hamiltonian paths and polygonalizations. We state these bounds as theorems in this section, and defer the proofs to later sections.

\begin{definition}
Define a \emph{non-crossing path} for a set of points $S$ to be a non-self-intersecting polygonal curve $P$ that passes through a subset of points of $S$, has endpoints in $S$, and turns only at points of $S$. Define the \emph{vertices} of a non-crossing path to be the set $P\cap S$, counting a point of $S$ as a vertex even when $P$ passes straight through that point. Define a \emph{non-crossing path sequence} to be the sequence of vertices in a non-crossing path of a given set of points in the plane, including also one-vertex sequences and the empty sequence. Let $|S|$ denote the number of points in $S$, let $\npath(S)$ denote the number of non-crossing paths with vertices in $S$, and let $\nham(S)$ denote the number of non-crossing Hamiltonian paths of~$S$.
\end{definition}

Each non-crossing path corresponds to two sequences (one for each end-to-end order in which its vertices can be placed). We can form a rooted tree with the non-crossing path sequences as its nodes, in which the empty sequence is the root, by defining the parent of any other sequence to be the subsequence obtained by removing its last vertex.
The non-crossing path sequences can be listed in polynomial time per sequence by performing a depth-first search of this tree. With a little care in listing the children of each node quickly, we can reduce this to linear time per sequence:

\begin{subroutine}[listing children of a non-crossing path sequence]
\label{alg:list-path-children}
To find the children of a non-crossing path sequence $\sigma$, ending at point $p$:
\begin{itemize}
\item Apply the simple stack-based linear time algorithm of Lee to determine the visibility polygon $V$ of the final vertex of the path, the region of the plane within which that vertex can be connected to another point by a segment that does not cross the existing path~\cite{Lee-CVGIP-83}.
\item Let $U$ be the set of points not in $\sigma$, and find the radial ordering of $U$ around $p$. When there are ties, keep only the closest of the tied points to $p$.
\item Merge the radial orderings of $U$ and $V$ to determine, for each point $u$ in $U$, the edge of $V$ that is crossed by a ray from $p$ through $u$.
\item Whenever the merge finds a point $u$ that is closer to $p$ than the corresponding edge of $V$, make a child sequence by concatenating $u$ to the end of $\sigma$.
\end{itemize}
The radially-sorted lists of all points around each of the given point can be precomputed and stored in $O(|S|^2)$ time; essentially, this is the same as the problem of constructing and storing an arrangement of lines dual to the points.
With this precomputation, all steps of this algorithm take time $O(|S|)$.
\end{subroutine}

\begin{subroutine}[listing all non-crossing paths]
\label{alg:all-paths}
To list all non-crossing paths of a given set of points, perform a depth-first search of the tree of non-crossing path sequences, using \cref{alg:list-path-children} to list the children of each node in the tree. For each node that the search reaches, output it as a path whenever the starting vertex of the sequence has a smaller index than the ending vertex, so that we only output each non-crossing path once. For a point set $S$, we spend $O(|S|^2)$ preprocessing time and $O(|S|)$ time per path. The number of paths is always $\Omega(|S|^2)$, even for collinear point sets, because a Hamiltonian path always exists and contains that many paths within it, so the preprocessing time is dominated by the per-path time and the total time is $O(|S|\cdot\npath(S))$.
\end{subroutine}

With these subroutines in hand, we can list all non-crossing Hamiltonian paths by the following very simple algorithm.

\begin{algorithm}[listing non-crossing Hamiltonian paths]
\label{alg:ham-paths}
To list all non-crossing Hamiltonian paths in a point set $S$:
\begin{itemize}
\item List all non-crossing paths by \cref{alg:all-paths}.
\item Whenever a path uses all points of $S$, output it as a non-crossing Hamiltonian path.
\end{itemize}
\end{algorithm}

\begin{theorem}
\label{thm:path-alg}
For a point set $S$ (not assumed to be in general position),
\cref{alg:ham-paths} takes time $\left(|S|\cdot\nham(S)\right)^{O(1)}$ to list all non-crossing Hamiltonian paths.
\end{theorem}

\begin{proof}
This follows from \cref{thm:path-equivalence}, later in this paper, which states that
\[\npath(S)=\bigl(|S|\cdot\nham(S)\bigr)^{O(1)}.\qedhere\]
\end{proof}

A very similar tree search can also be used for polygonalizations, instead of Hamiltonian paths.

\begin{definition}
Yamanaka et~al.~\cite{YamAviHor-DAM-21} define a \emph{surrounding polygon} of a point set $S$ to be a simple polygon having a subset of the points as its vertices, surrounding all of the vertices. These include the polygonalizations (in which the subset is all of the points) as well as other polygons; in particular, the convex hull is always a surrounding polygon. Define $\nsurround(S)$ to be the number of surrounding polygons in~$S$, and $\npoly(S)$ to be the number of polygonalizations in~$S$.
\end{definition}

Yamanaka et~al. define a tree structure on surrounding polygons in which the root is the convex hull, and the parent of any surrounding polygon is obtained by removing one vertex (in a canonically chosen way) from its cyclic sequence of vertices. It follows from a version of the two-ears theorem for polygons (often credited to G. H. Meisters, but used earlier by Max Dehn~\cite{Mei-AMM-75,Gug-AHES-77}) that every polygon that is not the convex hull has a parent.

\begin{subroutine}[listing surrounding polygons]
\label{alg:surrounding}
As Yamanaka et~al.~\cite{YamAviHor-DAM-21} describe, the surrounding polygons of  point set $S$ can be listed in time $|S|^{O(1)}$ per polygon, and space $O(|S|)$, by a depth-first search of the tree of polygons. The algorithm uses the method of \emph{reverse search}~\cite{AviFuk-DAM-96} to perform the depth-first search while only maintaining the identity of a bounded number of tree nodes.
\end{subroutine}

Again, we can use this to list all polygonalizations, as was already done by Yamanaka et~al.~\cite{YamAviHor-DAM-21}.

\begin{algorithm}[listing polygonalizations]
\label{alg:poly}
To list all polygonalizations of a point set~$S$:
\begin{itemize}
\item List all surrounding polygons by \cref{alg:surrounding}.
\item Whenever a polygon uses all points of $S$ output it as a polygonalization
\end{itemize}
\end{algorithm}

Yamanaka et~al. analyzed this algorithm as having singly-exponential time, but we instead prove that it is output-sensitive:

\begin{theorem}
\label{thm:cycle-alg}
For a point set $S$ (not assumed to be in general position),
\cref{alg:poly} takes time $\bigl(|S|\cdot\npoly(S)\bigr)^{O(1)}$ to list all polygonalizations.
\end{theorem}

\begin{proof}
This follows from \cref{thm:cycle-equivalence}, later in this paper, which states that
\[\nsurround(S)=\npoly(S)^{O(1)}.\qedhere\]
\end{proof}

\section{Counting paths}

The main result of this section is to prove a polynomial relation between the number of non-crossing paths and the number of non-crossing Hamiltonian paths, in any point set. This result combines a lower bound on the non-crossing Hamiltonian paths of a point set $S$, and an upper bound on the number of non-crossing paths, both expressed as a function of the following quantity.

\begin{definition}
Following~\cite{Epp-18}, define $\offline(S)$ to be the smallest $k$ such that removing $k$ points from $S$ leaves a collinear subset.
\end{definition}

It will be convenient to have the following standard bound on logarithms of binomial coefficients. We use $\log$ to denote the natural logarithm, but this choice of base makes no difference to the following lemma, because a change of base would only change the logarithm by a constant factor.

\begin{lemma}
\label{lem:log-binom}
For integers $k$ and $n$ with $k\le n/2$,
\[\log\binom{n}{k}=\Theta\left(k\log\frac{n}{k}\right).\]
\end{lemma}

\begin{proof}
By applying Stirling's formula, taking logarithms, and omitting terms that are $O(\log n)$, we obtain
\[\log\binom{n}{k}\approx \log\frac{n^n}{k^k(n-k)^{n-k}}=k\log\frac{n}{k}+(n-k)\log\frac{n}{n-k}.\]
This expression is symmetric in $k$ and $n-k$, and when $k\le n/2$, the first term of this approximation is the larger of the two.
\end{proof}

\subsection{Upper bound}

We begin our bounds on non-crossing paths with the simplest one to prove, the upper bound. It is known that the number of paths is at most singly-exponential in the number $n$ of vertices; we prove a tighter bound depending exponentially on the smaller number $\offline(S)$ and only polynomially on $n$. 

\begin{lemma}
\label{lem:path-upper}
Let $S$ be a set of $n$ points, with $\offline(S)=k$. Then
\[ \log\npath(S)=O\left(\log n + k\left(\log\frac{n}{k+1}\right)\right).\]
\end{lemma}

\begin{proof}
We describe a method for encoding a non-crossing path using this many bits of information, so that each path is uniquely described by this encoding. For an encoding with $b$ bits of information, the number of paths can be at most $2^b$ and the logarithm of this number is $O(b)$. Let $K$ be a subset of $S$, with $|K|=k$, such that $S\setminus K$ lies on a line $L$, and let $P$ be any given non-crossing path in $S$. Let $\ell=|L\cap S|$; because $k$ is the minimum number of points that can be removed to form a collinear set, no point of $K$ can lie on $L$, and $\ell=n-k$. To describe $P$, we combine the following pieces of information:
\begin{itemize}
\item The set $Q$ of points of $L\cap S$ that belong to $P$, but for which zero or one of their neighbors in the path belong to $L$. Each such point is one of the two neighbors of a point in $K$, or one of the two ends of $P$, so $|Q|\le 2k+2$. $Q$ can be encoded by specifying its size and the subset of $L\cap S$ of that size, out of $\tbinom{\ell}{|Q|}$ possibilities, so by \cref{lem:log-binom} the number of bits needed to specify it is $O\bigl(\log n + k+k\log(n/k)\bigr)$. (Both the $\log n$ term and the $+1$ in the statement of the lemma are included to handle the case when $k=0$ but $|Q|>0$. \cref{lem:log-binom} applies only when $k\le n/2$ but for larger $k$ the bound to be proven is superlinear and the result is immediate.)
\item For each point in $Q$, a specification of whether it has a neighbor in $L$, and if so in which direction. This takes $O(k)$ bits of information.
\item The induced subgraph $P[K\cup Q]$, a linear forest using only the points in $K\cup Q$, and omitting the edges of $P$ that lie entirely within $L$. As with any type of planar straight-line graph, the number of linear forests on $O(k)$ points is singly exponential in $k$~\cite{ShaShe-CPC-13}, so $P[K]$ can be encoded with $O(k)$ bits of information.
\end{itemize}
Then $P$ may be recovered by combining the induced subgraph $P[K\cup Q]$ with segments of $L$ starting and ending at points of $Q$ and continuing in the specified direction from each of these points. All pieces of this encoding add up to the stated bound on the number of bits needed to encode the entire path.
\end{proof}

\subsection{Visible-vertex paths}

\begin{figure}[t]
\centering\includegraphics[width=0.3\textwidth]{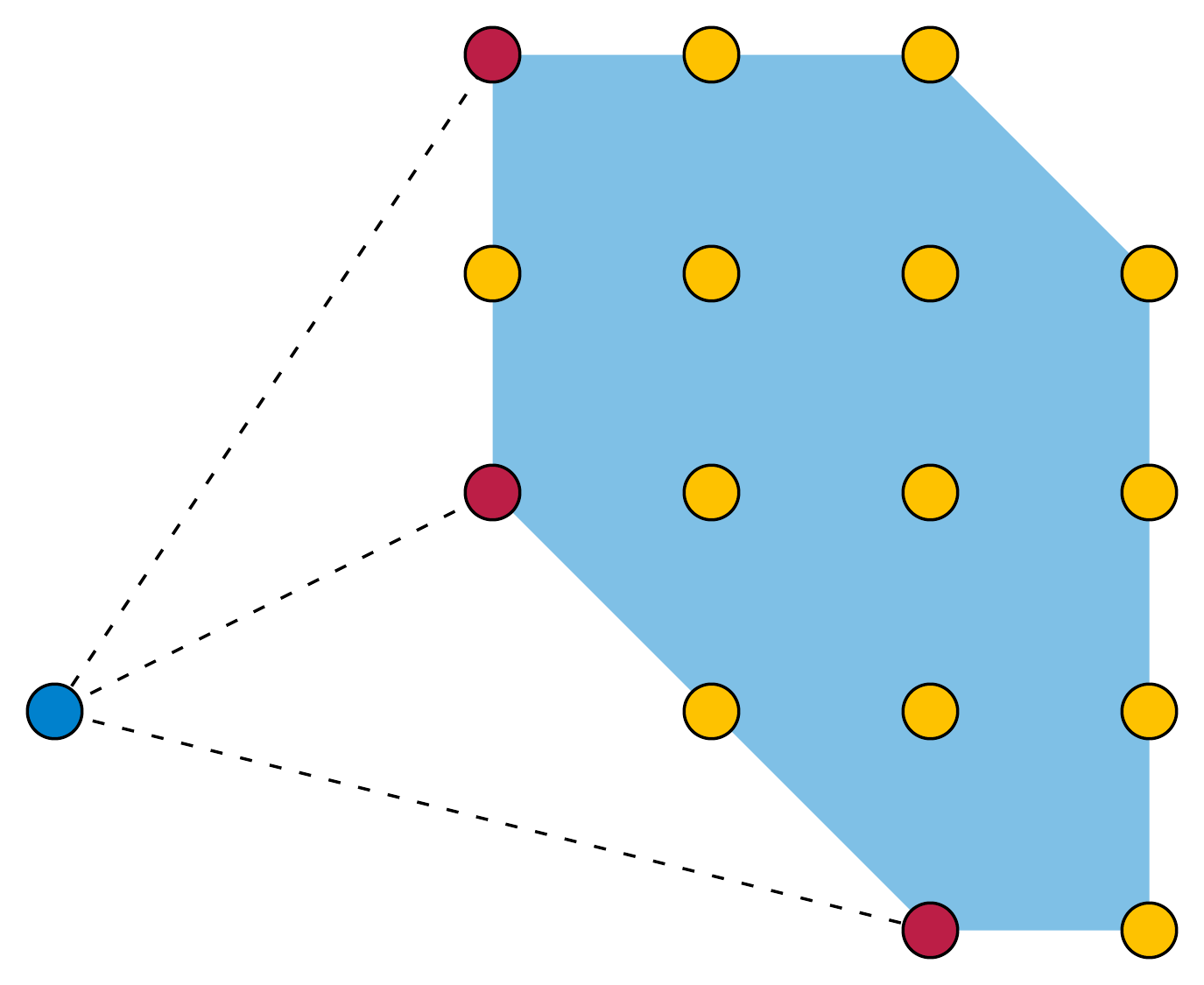} \qquad
\includegraphics[width=0.3\textwidth]{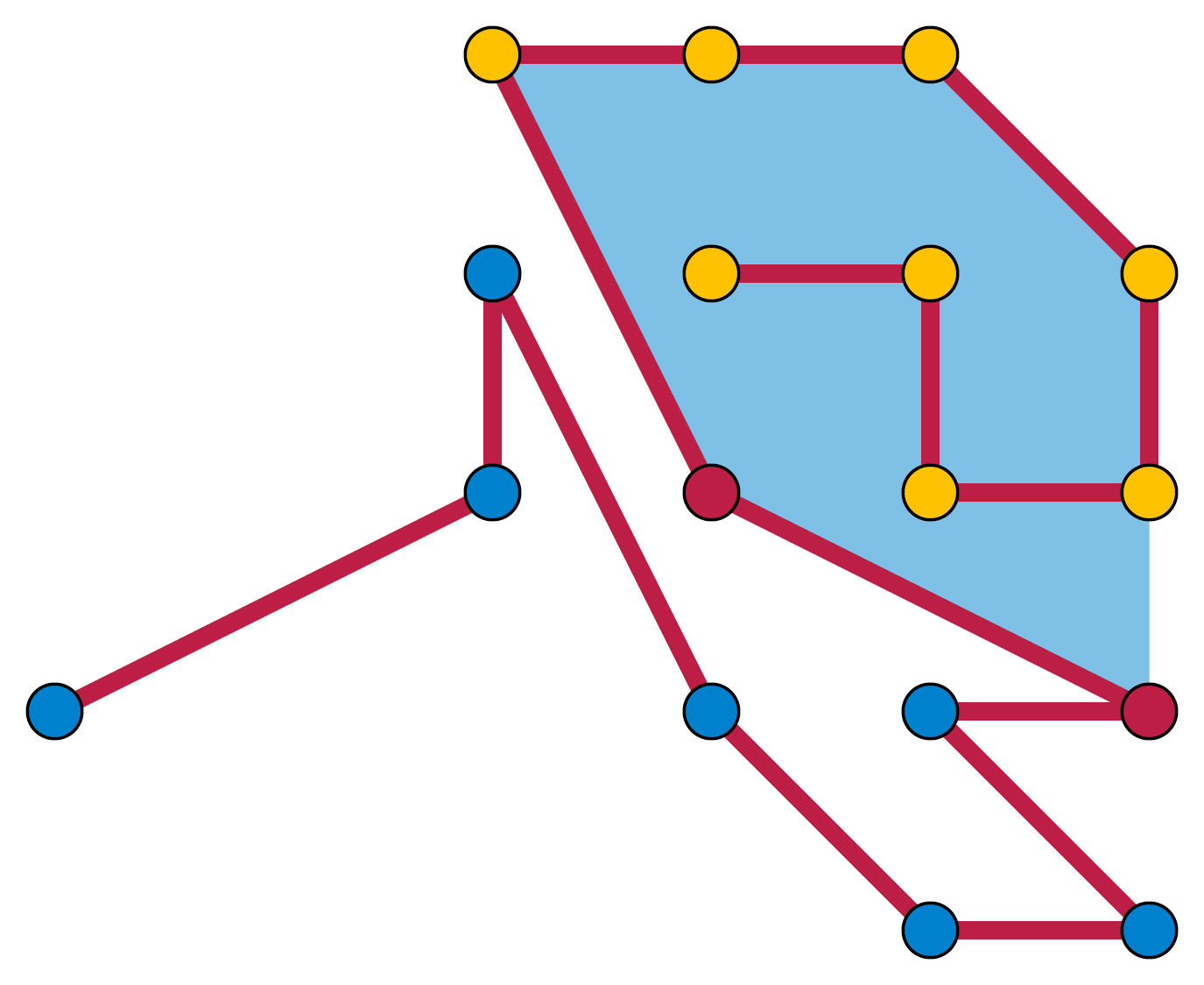}
\caption{Left: A point $p$ (blue), a point set $S$ (red and yellow), and the visible vertices of $S$ from $p$ (the three red vertices). Note that points of $S$ that lie within convex hull edges are not counted as vertices. Right: A maximal visible-vertex path (red edges) for $S\cup\{p\}$ starting from $p$, showing for one of its steps the hull (light blue) and visible vertices (red) of the remaining points.}
\label{fig:visible-vertices}
\end{figure}

Our lower bounds will greedily construct paths such that, at each step, the remaining unused points have a convex hull that is uncrossed by the current path and is visible from the endpoint of the current path.

\begin{definition}
Define a \emph{visible vertex} of a finite point set $S$, from a point $p$ that is not in the convex hull of $S$, to be a vertex $q$ of the convex hull of $S$ such that the open line segment $pq$ is disjoint from the convex hull of $S$ (\cref{fig:visible-vertices}, left). We do not consider points interior to edges of the convex hull to be vertices, and in particular they are not visible vertices, even when segment $pq$ is disjoint from the hull.
\end{definition}

\begin{observation}
\label{obs:visible}
Every nonempty finite point set $S$ and point $p$ not in the convex hull of $S$ has at least one visible vertex of $S$ from $p$. If $S\cup\{p\}$ is not collinear, there are at least two visible vertices.
\end{observation}

Define a \emph{visible-vertex path} to be a polygonal chain formed by a greedy algorithm that builds a sequence of vertices, beginning with any vertex of the convex hull, by repeatedly appending to the sequence a visible vertex of the remaining points not yet in the sequence, as viewed from the last point of the sequence (\cref{fig:visible-vertices}, right). A \emph{maximal visible-vertex path} is a visible-vertex path that uses all of the points in a given point set $S$; this name is justified by the following lemma.

\begin{lemma}
\label{lem:vvpath}
Every visible-vertex path is a non-crossing path. Every maximal visible-vertex path is a non-crossing Hamiltonian path. Every non-maximal visible-vertex path can be extended to a longer visible-vertex path.
\end{lemma}

\begin{proof}
Maximal visible-vertex paths are Hamiltonian, by definition: they are defined to be paths that use all the points. Because each step in a visible-vertex path is defined locally, without respect to the earlier parts of the path, the ability to extend every non-maximal path follows immediately from \cref{obs:visible}.

It remains to prove that the resulting paths are non-crossing. For every segment $pq$ of the path, the next segment is incident to $q$ and therefore cannot cross $pq$. All segments after the next segment lie within the convex hull of the remaining points after $q$. Since $p$ and $q$ are vertices of their respective convex hulls, the convex hull of the points after them in the path does not contain them. Moreover, segment $pq$ does not cross this hull, because if it did then $q$ would not be visible from $p$. Thus, these later segments are disjoint (as closed line segments) from the closed line segment $pq$ and so cannot cross $pq$ nor even pass through $q$. Therefore, $pq$ cannot intersect any later segment. For each pair of segments in the path, a crossing is ruled out by applying this argument to the earlier of the two segments in the path, so no crossing can exist.
\end{proof}

\begin{lemma}
\label{lem:2hull}
In every finite point set $S$, every two vertices $p$ and $q$ of the convex hull of $S$ form the endpoints of a non-crossing Hamiltonian path.
\end{lemma}

\begin{proof}
Form a maximal visible-vertex path beginning at $p$, at each step choosing a visible vertex that is not $q$ when this is possible.
By \cref{obs:visible}, the path will have two points to choose between (one of which is not $q$) until the remaining points become collinear. Once they do, all remaining points that are not $q$ will lie between $q$ and the current end of the path, so $q$ cannot be included in the path until it is the last point remaining.
\end{proof}

\subsection{Lower bound for far-from-collinear sets}

To lower-bound the number of non-crossing Hamiltonian paths in a point set $S$, as a function of $\offline(S)$, we divide into two cases: the case where $\offline(S)$ is large (at least proportional to a constant fraction of $|S|$) and the case where it is small. The following lemma is valid for all non-collinear $S$, but is only tight (within a constant factor of the logarithm) in the large case.

\begin{figure}[t]
\centering\includegraphics[width=0.9\textwidth]{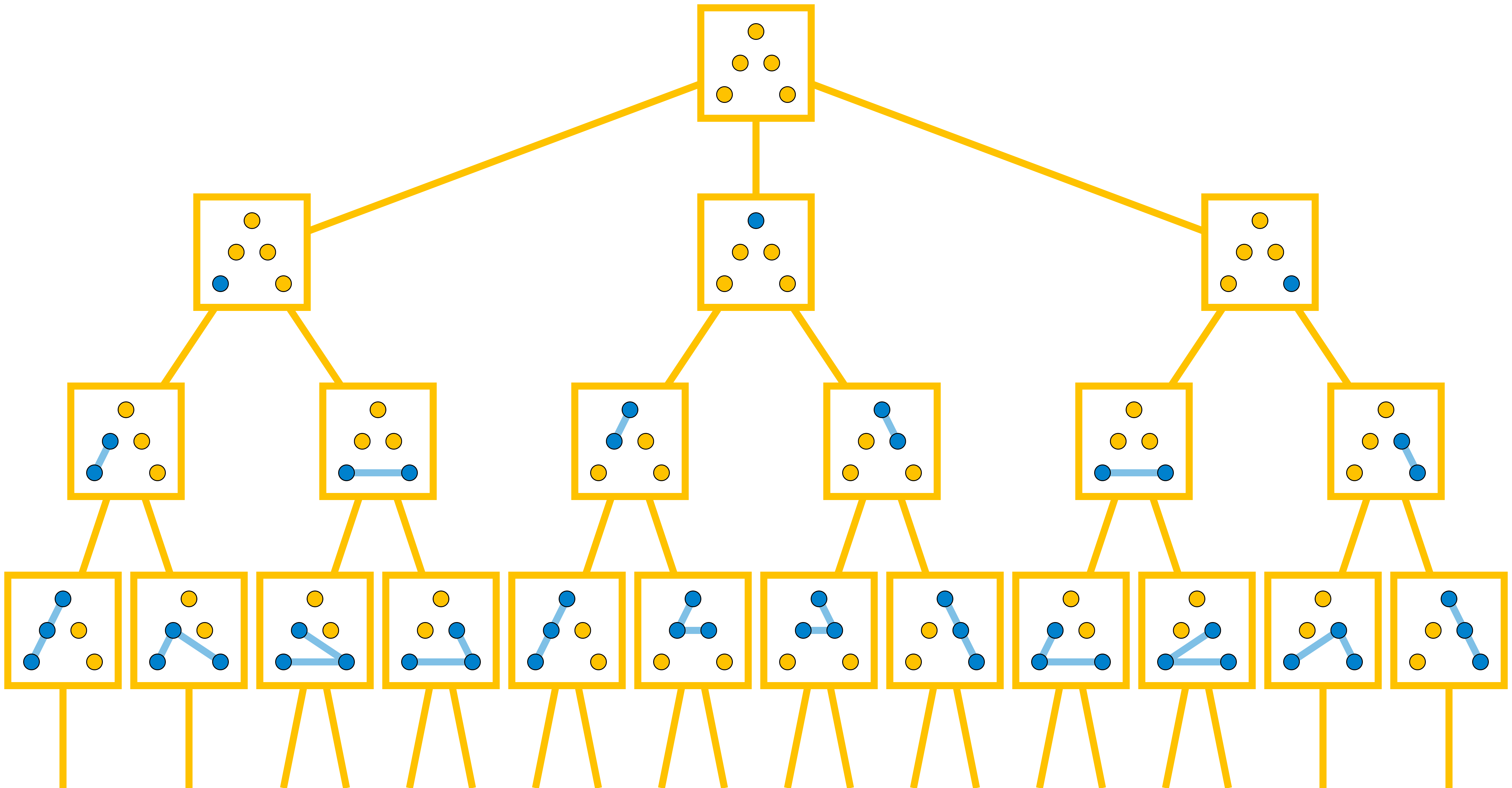}
\caption{A tree of the visible-vertex paths in a point set, where the parent of each path is obtained by removing its last point. This point set has $\offline(S)=2$. The root and the next two levels of nodes each have multiple children, but some nodes on the last level shown in the figure have only one child, because the last point on the path is collinear with all remaining points.}
\label{fig:tree-of-paths}
\end{figure}

\begin{lemma}
\label{lem:count-vv}
For every point set $S$ that does not lie on a single line, the number of non-crossing Hamiltonian paths is at least $\tfrac32\cdot 2^{\offline(S)}.$
\end{lemma}

\begin{proof}
We form a tree $T$ of visible-vertex paths, in which the root is the empty sequence of vertices and the parent of any nonempty path is obtained by removing its last vertex (\cref{fig:tree-of-paths}). By \cref{lem:vvpath}, the leaves of $T$ are non-crossing Hamiltonian paths. The root node of $T$ has at least three children (one for each convex hull vertex).
In the nodes of $T$ at distance at most $\offline(S)$ from the root, the last point in the path represented by each node is not collinear with the remaining points, by the definition of $\offline$. It follows by \cref{obs:visible} that these nodes have at least two choices of visible vertices and therefore that they have at least two children.

As a tree in which the root node has at least three children and the nodes in the next $\offline(S)$ levels have at least two children, $T$ has at least $3\cdot 2^{\offline(S)}$ leaves. Each leaf is a non-crossing Hamiltonian path, and each non-crossing Hamiltonian path can come from at most two leaves (one for each endpoint, if both endpoints are convex hull vertices). Therefore, there are at least $\tfrac32\cdot 2^{\offline(S)}$ non-crossing Hamiltonian paths.
\end{proof}

\subsection{Lower bound for near-collinear sets}

Although the next lemma covers only a special class of point sets, it is the key to our lower bounds for the case where $\offline(S)$ is small.

\begin{lemma}
\label{lem:1side}
Let $S$ be a set of points with $|S|=n$, such that $\ell$ points of $S$ lie on a line $L$ and the remaining points all lie on the same side of the line. Then
\[
\nham(S)\ge\binom{n-\lceil \ell/2\rceil}{\lfloor \ell/2\rfloor}.
\]
\end{lemma}

\begin{figure}[t]
\centering\includegraphics[scale=0.35]{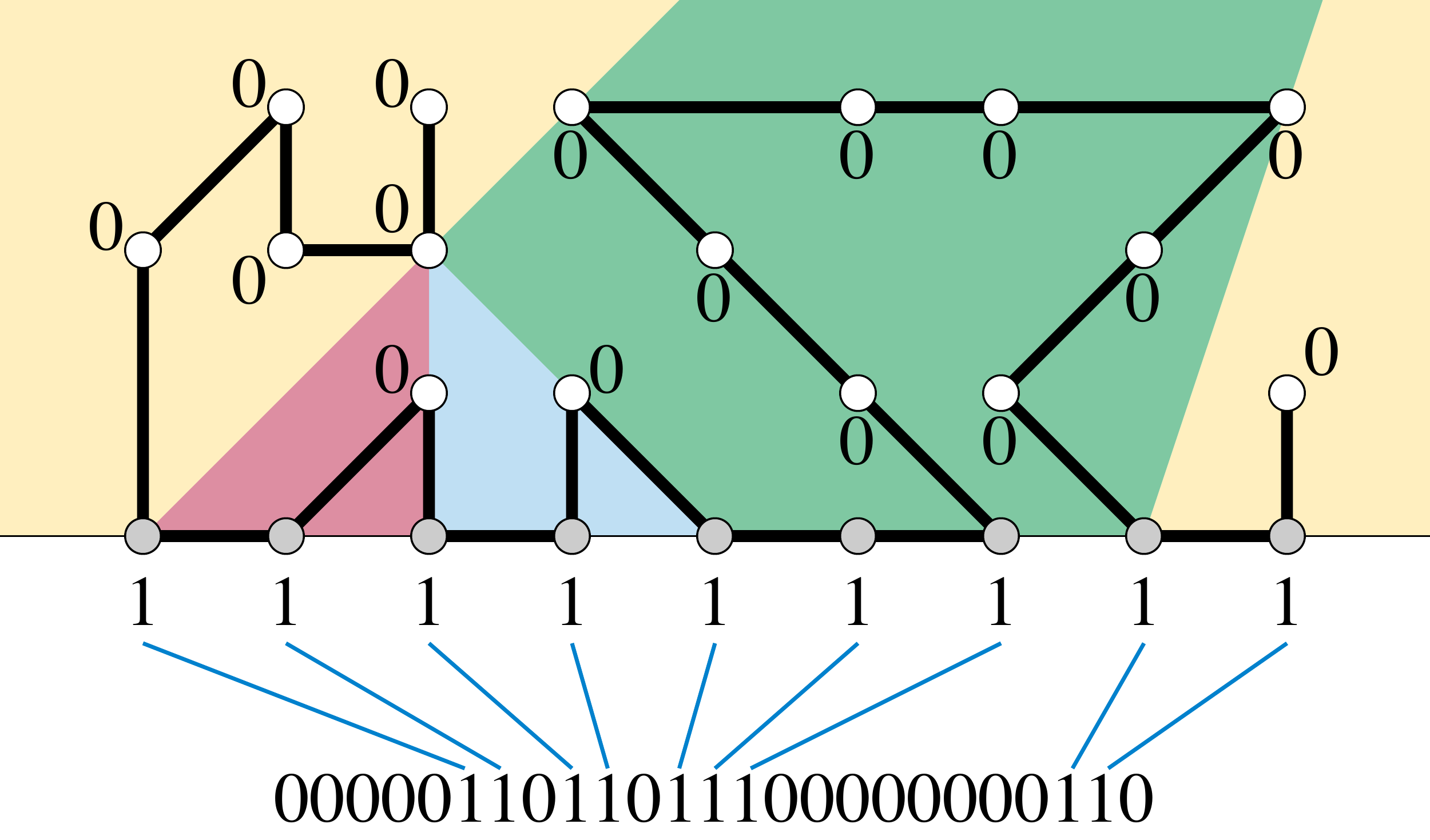}
\caption{Partition of the halfspace above $L$ into convex subsets, a non-crossing Hamiltonian path respecting the partition, and its signature, a $010$-avoiding binary sequence.}
\label{fig:signature}
\end{figure}

\begin{proof}
We may assume without loss of generality that $\ell\ge 2$, for otherwise the lemma states only that there exists a single non-crossing Hamiltonian path, known to be true.
For convenience consider an orientation of the plane in which $L$ is horizontal and $S$ lies in the closed halfspace above $L$.
We will prove the lemma by constructing many non-crossing Hamiltonian paths in which the points of $L\cap S$ appear in left-to-right order, and no point in $L$ has two neighbors in $S\setminus L$. For any such path, define its \emph{signature} to be the binary sequence with a $1$-bit for points in $L$ and a $0$-bit for points in $S\setminus L$ (\cref{fig:signature}). It has length $n$, with $\ell$ $1$-bits, and does not have any three consecutive bits in the pattern $010$. The number of $010$-avoiding sequences with $n$ bits and $\ell$ $1$-bits is~\cite{A180562}
\[
\sum_{j=0}^{n-\ell} (-1)^j \binom{n-\ell-1}{j}\binom{|n-2j|}{\ell-j}\ge \binom{n-\lceil \ell/2\rceil}{\lfloor \ell/2\rfloor},
\]
where for the simpler formula on the right hand side of the inequality the $1$-bits are grouped into $\lfloor \ell/2\rfloor$ pairs (with one group of three if $\ell$ is odd) and we count only the sequences in which each pair or triple appears consecutively. (By the assumption that $\ell\ge 2$, at least one such grouping is possible.) As we argue in the remainder of the proof, every $010$-avoiding sequence is the signature of at least one non-crossing Hamiltonian path, so this lower bound on the number of $010$-avoiding sequences also provides a lower bound on the number of non-crossing Hamiltonian paths.

For a given $010$-avoiding sequence $\sigma$, let $n_i$ denote the length of the $i$th non-empty block of consecutive $0$-bits in $\sigma$. We will partition the halfplane above $L$ into convex sets $C_i$, each containing $n_i$ points of $S\setminus L$, by 
a greedy process that maintains a convex subset of the halfplane containing the remaining points to be partitioned. Initially the convex subset is the entire halfplane and the remaining points are $S\setminus L$. On the $i$th step (for any $i$ other than the last one), let $p_i$ be the point of $S\cap L$ that corresponds to the $1$-bit of $\sigma$ following the $i$th block of $0$-bits.
Sort the remaining points of $S\setminus L$ radially around $p_i$ (in left to right order with respect to $L$) breaking ties in favor of closer points to $p_i$, and let $q_i$ be the point in position $n_i$ of this sorted order. Draw line $p_iq_i$, separating $C_i$ on its left from a remaining convex subset on its right. Assign the $n_i$ points up to $q_i$ in the radial sorted order to set $C_i$, and leave the remaining points (possibly including farther points on line $p_iq_i$) unassigned. In the final step of this construction, assign all remaining points to the final remaining convex region. For instance, in \cref{fig:signature}, the  leftmost set $C_1$ (yellow) is separated from the rest of the halfplane by line $p_1q_1$. Here $p_1$ is the leftmost point of $L$, and $q_1$ is the fifth point in the radial ordering around $p_1$. Point $q_1$ lies on the boundary of four convex regions but is assigned to the first,~$C_1$. Line $p_1q_1$ also contains another point of $S$, farther from $p_1$, which is assigned to $C_4$ (green).

Once this partition into convex sets has been determined, use \cref{lem:2hull} to find a non-crossing Hamiltonian path within each convex set $C_i$ that starts and ends at its (one or two) points on $L$, and connect these paths in sequence by segments of $L$ to form a non-crossing Hamiltonian path for all of $S$, with the given $010$-avoiding sequence $\sigma$ as its signature.
\end{proof}

\subsection{Putting the bounds together}

Combining the two different lower bounds into a single formula, we have:

\begin{lemma}
\label{lem:path-lower}
Let $S$ be a set of $n$ points with $\offline(S)=k.$ Then
\[ \log\nham(S)=\Omega\left(k\left(\log\frac{n}{k+1}\right)\right).\]
\end{lemma}

\begin{proof}
For any $k$, the logarithm of the number of non-crossing Hamiltonian paths is $\Omega(k)$ by \cref{lem:count-vv}.
For $k\ge n/3$, $\log(n/k)=O(1)$, so for this range of $k$, this $\Omega(k)$ bound is equivalent to the bound stated in the theorem.

For smaller values of $k$, let $K$ be any set of $k$ points whose removal from $S$ leaves a collinear set of size $\ell=n-k=\Omega(n)$, belonging to a line $L$. Partition $K$ into the two subsets $K_1$ and $K_2$ on the two sides of $L$, with $|K_1|\ge |K_2|$; let $|K_1|=k'\ge k/2$.
Let $p$ be the first point of $S\cap L$, in the sorted sequence of the points along this line, let $\ell'=\ell-1$ be the number of remaining points in $S\cap L$, and let $n'=\ell'+k'$. By \cref{lem:1side}, the number of non-crossing Hamiltonian paths of $S\setminus K_2$ that start at $p$ is at least
\[
\binom{n'-\lceil\ell'/2\rceil}{\lfloor\ell'/2\rfloor}=\binom{\lfloor\ell'/2\rfloor+k'}{k'}.
\]
Each such path can be extended to a non-crossing Hamiltonian path of all of $S$ by concatenating any non-crossing path through $p$ and the points of $K_2$. By the assumption that $k<n/3$, the bottom term of the right binomial coefficient is at most half the top term, allowing us to apply \cref{lem:log-binom}. By this lemma, and the facts that $k'=\Theta(k)$ and $\ell'=\Theta(n)$, the logarithm of this binomial coefficient is $\Omega\bigl(k\log(n/k)\bigr)$ as stated.
\end{proof}

Although the upper bound of \cref{lem:path-upper} and the lower bound of \cref{lem:path-lower}  are not quite the same, we can combine them to achieve a constant factor approximation to the logarithm of the number of non-crossing paths, or Hamiltonian paths.

\begin{theorem}
\label{thm:path-equivalence}
For a given point set $S$,
\[ \npath(S)=\bigl(|S|\cdot\nham(S)\bigr)^{O(1)}.\]
\end{theorem}

\begin{proof}
Taking logs of both sides, it is equivalent to write that
\[\log\npath(S)=O(\log |S|+\log\nham(S)).\]
This follows immediately from \cref{lem:path-lower} and \cref{lem:path-upper}, according to which $\log\nham(S)$ is lower-bounded and $\log\npath(S)$ upper-bounded (respectively) to within constant factors by formulas that differ from each other only in an additive $\log|S|$ term.
\end{proof}

\section{Counting cycles}
\label{sec:counting-cycles}

For counting both surrounding cycles and polygonalizations of general-position point sets, in place of $\offline(S)$ (which the general-position assumption makes trivial) we use the following parameter:

\begin{definition}
Let $\inhull(S)$ denote the number of points of $S$ that are interior to the convex hull of $S$.
\end{definition}

For counting cycles and polygonalizations of point sets that are not assumed to be in general position, we will use a combined analysis in terms of both $\offline$ and $\inhull$.

\iffull
\subsection{Upper bound}
\fi

Our bounds for surrounding cycles and polygonalizations follow similar arguments to our bounds for paths.
\iffull
\cref{lem:path-upper}. A key observation (true for surrounding cycles but not paths) is that the vertices of the convex hull that do not connect to interior points behave predictably:
\fi

\iffull
\begin{observation}
In a surrounding cycle for a point set $S$, every vertex of the convex hull of $S$ is a vertex of the cycle, and every edge of the cycle that has two convex hull vertices as its endpoints must be a convex hull edge.
\end{observation}

\begin{proof}
The convex hull vertices cannot be interior to the hull or interior to an edge, by definition, so being a vertex of the cycle is the only way they can be surrounded. Any other segment between hull vertices would block all visibilities from one side of the segment to the other. A non-crossing cycle using such a segment could only have vertices on one side of the segment, and would be unable to have among its vertices the vertices of the convex hull of $S$ that lie on the other side.
\end{proof}
\fi

\begin{lemma}
\label{lem:cycle-upper}
Let $S$ be a set of $n$ points, with $\inhull(S)=h$. Then
\[ \log\nsurround(S)=O\left(h\left(\log\frac{n}{h}+1\right)\right).\]
\end{lemma}

\iffull
\begin{proof}
As in \cref{lem:path-upper} we describe a method for encoding a surrounding cycle using this many bits of information, so that each cycle is uniquely described by this encoding. Let $H$ be the set of convex hull vertices of $S$, let $I=S\setminus H$ be the set of interior points, and let $C$ be the surrounding cycle that we are trying to describe. To describe $C$, we combine the following pieces of information:
\begin{itemize}
\item The set $Q$ of points of $H$ that belong to $C$, but for which zero or one of their neighbors in the hull belong to $C$. There can be at most $2h$ such points (two neighbors of each point in~$I$). $Q$ can be encoded by specifying its size and the subset of $H$ of that size, out of $\tbinom{n-h}{|Q|}$ possibilities, so by \cref{lem:log-binom} the number of bits needed to specify it is $O\bigl(h+h\log(n/h)\bigr)$. (The case where $|Q|$ is too large for \cref{lem:log-binom} to apply is trivial as this would make the bound we are proving superlinear.)
\item For each point in $Q$, a specification of whether it has a neighbor in $H$, and if so in which direction. This takes $O(h)$ bits of information. (This information is redundant when the point in $Q$ is adjacent in $H$ to a point of $H\setminus Q$, as all such pairs must be neighbors, but is needed when $Q$ contains multiple consecutive points of $H$.)
\item The induced subgraph $C[I\cup Q]$, a cycle or linear forest using only the points in $I\cup Q$. As with any type of planar straight-line graph, the number of linear forests on $O(h)$ points is singly exponential in $h$~\cite{ShaShe-CPC-13}, so $C[I\cup Q]$ can be encoded with $O(h)$ bits of information.
\end{itemize}
Then $C$ may be recovered by combining the induced subgraph $P[I\cup Q]$ with convex hull subsequences starting and ending at points of $Q$ and continuing in the specified direction from each of these points. All pieces of this encoding add up to the stated bound on the number of bits needed to encode the entire path.
\end{proof}
\fi

\begin{corollary}
\label{cor:cycle-upper}
Let $S$ be a set of $n$ points with $\min(\offline(S),\inhull(S))=m$. Then
\[ \log\nsurround(S)=O\left(m\left(\log\frac{n}{m}+1\right)\right).\]
\end{corollary}

\iffull
\begin{proof}
When $m=\inhull(S)$ this is \cref{lem:cycle-upper}. When $m=\offline(S)$ this follows immediately from \cref{lem:path-upper}, as each surrounding cycle can be made into a path by removing an edge, and no two such paths can come from the same cycle.
\end{proof}
\fi

\iffull
\subsection{Visible-vertex paths revisited}

\cref{lem:2hull} proved the existence of a non-crossing Hamiltonian path connecting two vertices of a convex hull, using a construction based on visible-vertex paths. For our lower bound on polygonalizations, we will need many such paths. As a warm-up, we provide a version of the bound that we need that applies to point sets in general position.

\begin{figure}[t]
\centering\includegraphics[scale=0.35]{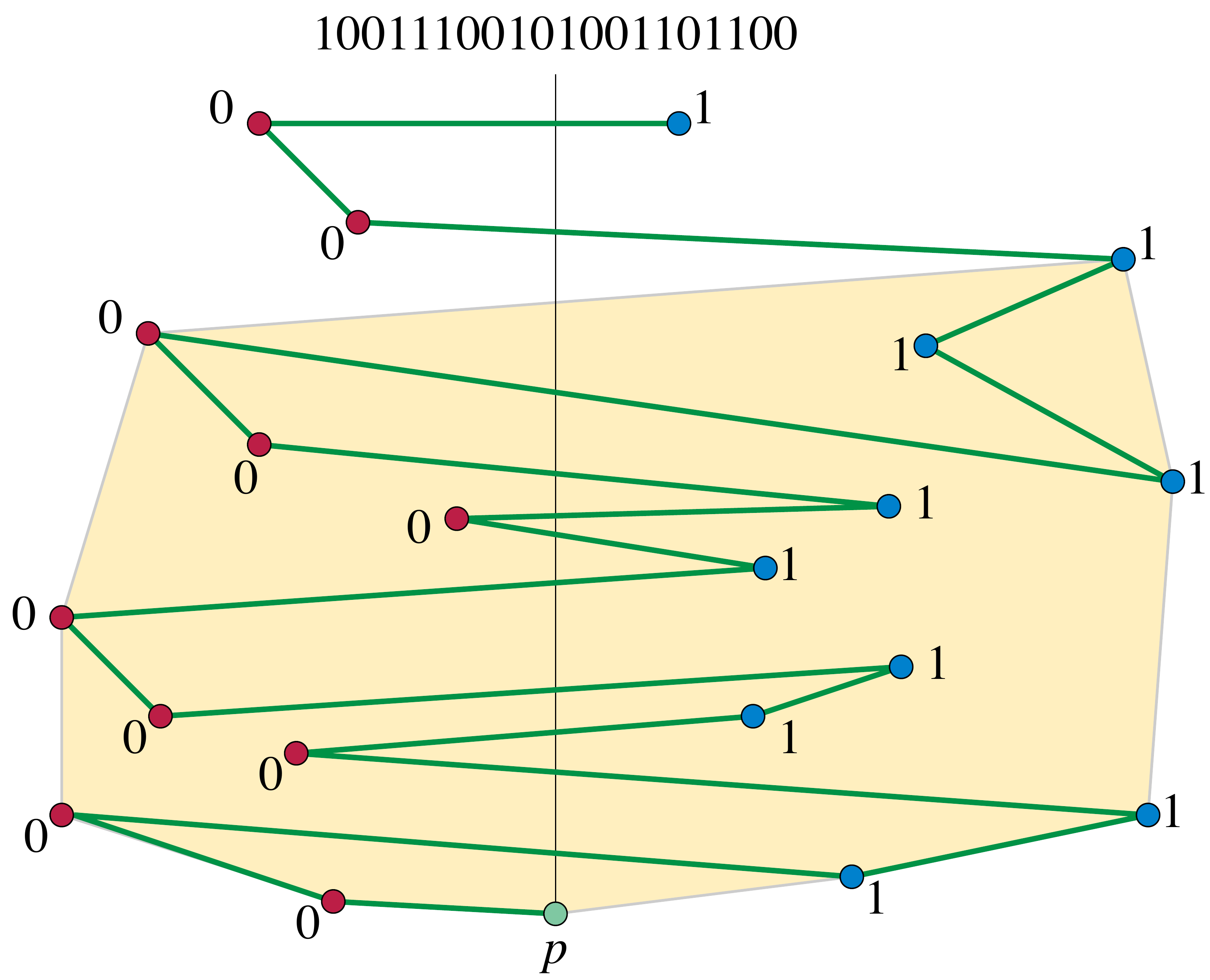}
\caption{Illustration for the proof of \cref{lem:distinct-paths}: The green vertex-visible path is labeled by the sequence $\sigma=10011100101001101100$ for the points shown, split by the vertical line through $p$ (the bottom green point) into $A$ (the red points left of the line) and $B$ (the blue points right of the line). One of the convex hulls of the remaining points partway through the sequence is shown; the path reaches this hull at one endpoint of its unique edge with endpoints in both $A$ and $B$.}
\label{fig:distinct-paths}
\end{figure}

\begin{lemma}
\label{lem:distinct-paths}
Let $S$ be a set of $n$ points in general position, and let $p$ be any vertex of the convex hull of $S$. Then there are at least $\tbinom{n-1}{\lfloor (n-1)/2\rfloor}$ distinct vertex-visible paths that start at other vertices of the convex hull of $S$ and end at $p$.
\end{lemma}

\begin{proof}
We assume $n\ge 3$ else the lemma is trivial. Draw a line $L$ through $p$ that does not pass through any other point of $S$, splitting $S\setminus\{p\}$ into two subsets $A$ and $B$ in its two half-planes, with $|A|=\lfloor (n-1)/2\rfloor$ and $|B|=\lceil (n-1)/2\rceil$. For each of the possible $\tbinom{n-1}{\lfloor (n-1)/2\rfloor}$ binary sequences $\sigma$ having $|A|$ 0-bits and $|B|$ 1-bits, we will find a vertex-visible path $P_\sigma$ that has a point of $A$ for each 0 in $\sigma$ and a point of $B$ for each 1 in $\sigma$, ending in $p$.

To do so, we will construct this sequence one point at a time. As long as the remaining unchosen points include at least one point of $A$ and at least one point of $B$, the line $L$ will continue to separate points in $A$ from points in $B$. The convex hull of the remaining points must have a ``bridge'' edge $e$ with one point in each set; $e$ is crossed by $L$. Throughout the process we will maintain the following property as an invariant: the sequence constructed so far, both endpoints of this bridge edge $e$ are valid additions to a vertex-visible path (although only one of the two will match sequence $\sigma$). This invariant is  automatically true at the start of the path construction process, when the path constructed so far is empty, because any convex hull vertex is a valid addition to an empty vertex-visible path.

At each subsequent step, let $e$ be the bridge edge. By the invariant, both endpoints of $e$ are valid additions to the path. We choose the endpoint that belongs to the set determined by the corresponding bit of $\sigma$, $A$ for a 0-bit or $B$ for a 1-bit. By the general position assumption, this chosen endpoint can see the other endpoint of $e$, in the other set. If it is not the last point in the set ($A$ or $B$) that contains it, it can also see at least one vertex of the convex hull of the remaining points that lies in its own set. Since the convex hull vertices that it sees must form a contiguous subsequence of the convex hull, it can see both endpoints of the new bridge edge crossed by $L$, separating $A$ from $B$, and the invariant is maintained.

Once all vertices in $A$ or in $B$ have been included in the vertex-visible path, we can complete the path in the same way as \cref{lem:2hull}. Because all paths constructed in this way can be distinguished from each other by the order in which they use vertices from $A$ and from $B$, they are all distinct from each other.
\end{proof}

\cref{fig:distinct-paths} illustrates this construction. The full bound that we need uses similar ideas, but must handle point sets that are not in general position. In particular we must handle the case that, after partially constructing a path as in the proof of \cref{lem:distinct-paths}, we reach a situation in which the edge of the remaining convex hull, crossed by $L$, contains many points, so that its two endpoints are not visible to each other and cannot be made visible by removing only a small number of points. To handle this case, we will apply \cref{lem:1side} (which finds many paths through a point set with many collinear points) to the collinear points on this convex hull edge.

\begin{lemma}
\label{lem:nongen}
Let $S$ be a set of points such that at most $|S|/6$ points lie on any line.
Let $p$ be any vertex of the convex hull of $S$. Then the number of non-crossing Hamiltonian paths through $S$, starting at another convex hull vertex and ending at $p$, is at least singly exponential in $|S|$.
\end{lemma}

\begin{proof}
The process in \cref{lem:distinct-paths} can be reinterpreted as non-deterministically choosing either a point from $A$ or a point from $B$ at each step, until running out of points in one of the two sets and then completing the path deterministically. This produces a binary tree of possible choices, and the the proof of \cref{lem:distinct-paths} can be interpreted as counting the branches of this tree. A simpler analysis, less tight but still exponential, would observe that in this tree, each choice uses up one point of either $A$ or $B$, and the number of points in both steps starts at approximately $n/2$. Therefore, the height (minimum number of steps from the root of the decision tree in any branch until running out of choices) is approximately $|S|/2$, and the number of branches is at least $2^{|S|/2}$. We will generalize this process to allow non-binary choices, using a weighted definition of the height of the decision tree where the weight of each step is the binary logarithm of the number of available choices at that step. If we can show that each step that uses up some number $x$ of points gives us a number of choices that is exponential in $x$, then the weight of that step will be proportional to $x$ and the total height of the decision tree will be again linear in $|S|$, implying that again it has an exponential number of branches.
With these generalities out of the way, it remains to show how to generalize the non-deterministic decision process of \cref{lem:distinct-paths} so that, whenever we use up many points, we select from a number of choices that is exponential in the number of used-up points.

\begin{figure}[t]
\centering\includegraphics[scale=0.35]{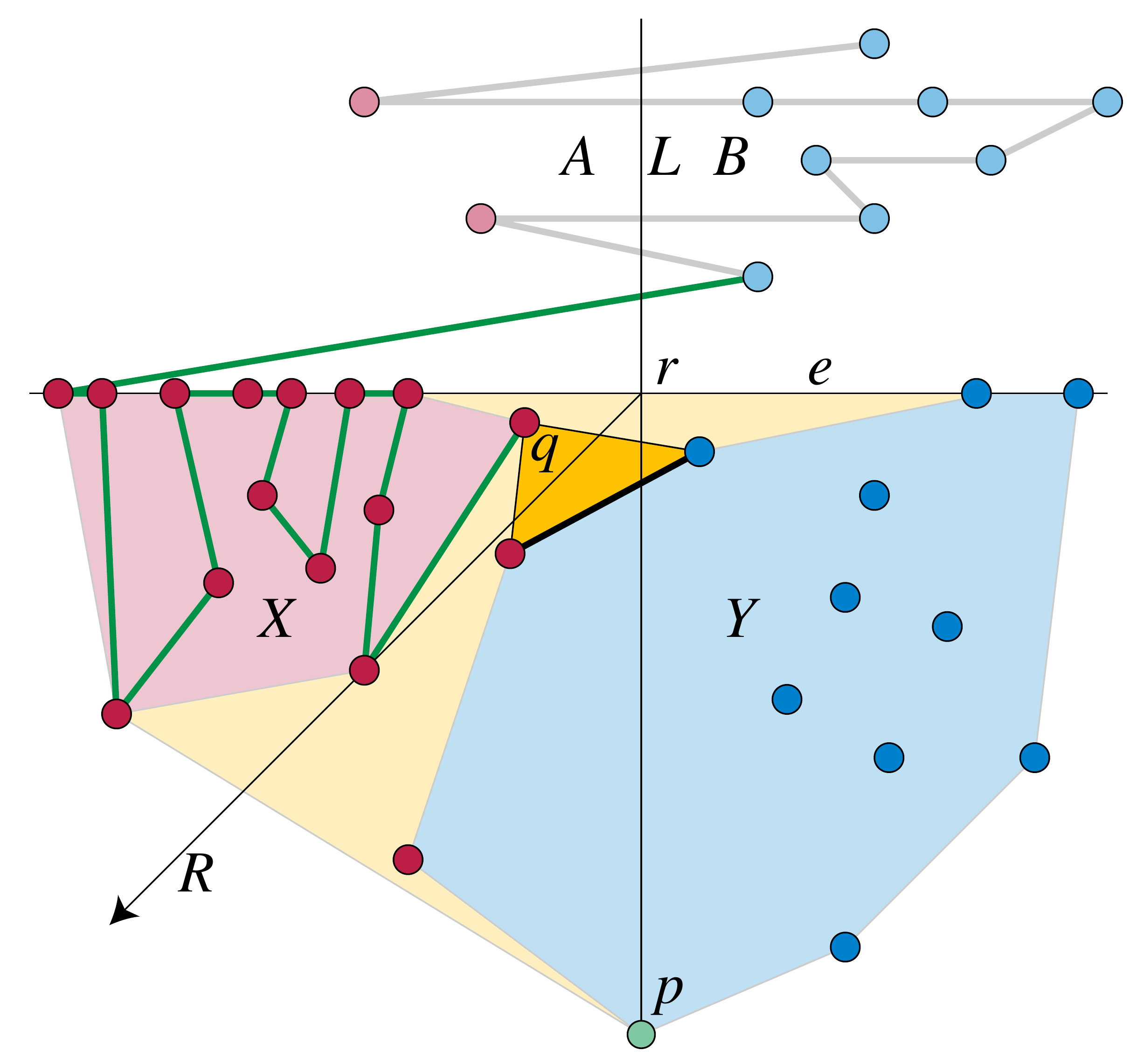}
\caption{The case of \cref{lem:nongen} when there are many points on edge $e$.}
\label{fig:nongen}
\end{figure}

We define the line $L$ and the subsets of points $A$ and $B$, initially of size approximately $|S|/2$, in exactly the same way as in the proof of \cref{lem:distinct-paths}. We will build a non-crossing path to $p$, covering all points in $S$, in a sequence of steps. It will not necessarily be a visible-vertex path, but we will maintain after each step the same invariants as in \cref{lem:distinct-paths}. Namely, the path constructed so far will remain disjoint from the convex hull of the remaining points (so that any non-crossing path within those remaining points will be non-crossing globally), and as long as $A$ and $B$ are both non-empty, the unique $A$--$B$ edge $e$ of the convex hull of the remaining points will be entirely visible from the last point that has already been added to the path, so that any point on $e$ may be added as the next point on the path. We distinguish the following cases:
\begin{itemize}
\item If there are fewer than $|S|/3$ points remaining in $A$ or $B$, stop making choices and complete the rest of the path deterministically.
\item If $A\cap e$ and $B\cap e$ are both small (at most three points), we extend the current path by a segment to one of the two endpoints of $e$, and then (if its visibility along $e$ is blocked by another point in the same set) by more segments to one or both blocking points. This case produces two choices (which endpoint of $e$ to extend to) and uses up at most three points of $A$ or $B$. It maintains the invariant by the same reasoning as in \cref{lem:distinct-paths}.
\item Otherwise, at least one of $A\cap e$ or $B\cap e$ is large (at least four points). We describe the case when $A\cap e$ is large; the case for $B\cap e$ large is symmetric. Let $x=|A\cap e|$, the number of points of $A$ on edge $e$ of the convex hull. Let $r$ be the point where $L$ crosses $e$. Draw a ray $R$ from $r$ that separates the remaining points not yet on the path into two subsets $X$ and $Y$ with $(A\cap e)\subset X$ and $|X|=2x$, breaking ties in such a way that the convex hulls of $X$ and $Y$ stay disjoint. Because $x\le |S|/6$ points lie on a line and at least $|S|/3$ remaining points lie in $A$, $X\subset A$, and ray $R$ lies entirely on $A$'s side of line $L$. In \cref{fig:nongen}, $e$, $r$, and $R$ are labeled; $x=7$, and $X$ is the subset of 14 points whose convex hull is shown as the red shaded area. The blue shaded area is the convex hull of $Y$.

Among the vertices of the convex hull of $X$, at least one (the point in $A\cap e$ closest to $r$) can see a point in $B\cap Y$, the nearest one on edge $e$, along a segment exterior to the hull. Also, at least one of the vertices of the convex hull of $X$ on $R$ can see any remaining points in $A\cap Y$, along a segment exterior to the hull. Therefore, some vertex $q$ of the convex hull of $X$, on or between these two hull vertices, can see all of the unique $A$--$B$ edge of the convex hull of $Y$. For instance, $q$ may be chosen by triangulating the region between the convex hulls of $X$ and $Y$, and choosing the apex of the triangle on this edge. \cref{fig:nongen} depicts a case where $q$ can neither be on $e$ nor $R$.

We will extend the current path, from its current endpoint, to one of the two endpoints of $X\cap e$ (in the figure, the endpoint farthest from $r$), through all points of $X$, ending at $q$. Because $e$ separates $Y$ from the first of these added segments, the remaining added segments lie within the convex hull of $X$, and $R$ separates the convex hulls of $X$ and $Y$, it follows that this extension maintains the invariant that the path so far is disjoint from the convex hull of the remaining point set $Y$. Because $q$ was chosen as a point that could see the $A$--$B$ edge of the convex hull of $Y$, it follows that this extension also maintains the invariant that any point along this edge could be added to the path next.

By \cref{lem:1side}, the number of paths through $X$ that start and end at the two extreme points of $X\cap e$ is exponential in $x$. However, instead we need paths that start at one extreme point of $X\cap e$ and end at $q$. To obtain an exponential bound on the number of paths of this type, we use the same method as in the proof of \cref{lem:1side}: we consider $010$-avoiding sequences of length $2x$, in which each $0$ describes a point on the path that belongs to $X\cap e$ and each $1$ describes a point on the path that belongs to $X\setminus e$. Each such sequence can then be realized as a path by the method used in \cref{lem:1side}, which creates a convex subset of $X\setminus e$ corresponding to each block of consecutive 1-bits of the sequence. In order to obtain a path that starts at an extreme point of $X\cap e$ and that ends at $q$, it is necessary and sufficient that the corresponding $010$-avoiding sequence have a $0$ at one of its ends (representing the starting extreme point) and a large enough consecutive block of 1's at the other end for $q$ to be included in that block. Here ``large enough'' depends on the arrangement of $X$, but each point of $X\setminus e$ contributes to this required block size only when it is more extreme than $q$ in the radial ordering of points around one extreme point of $X\cap e$, and when this happens it cannot also contribute to the required block size for the other extreme point. Therefore, for one of the extreme points of $X\cap e$, at least $x/2$ of the $1$-bits are free to be anywhere in the $010$-avoiding sequence, rather than forced to be in the contiguous block of $1$s at the end of the sequence. This is enough to produce an exponential number of $010$-avoiding sequences and an exponential number of path extensions ending at $q$.
\end{itemize}
Each step produces a number of choices exponential in the number of points that it uses up, until we have used up at least $|S|/6$ points and switch to the deterministic completion. Therefore, the total number of branches of the decision tree, and the total number of paths that it can produce is exponential in $|S|$.
\end{proof}
\fi

\iffull
\subsection{Lower bound for far-from-convex sets}

As we now show, for point sets that have many points interior to their convex hull, and few points collinear, the number of polygonalizations must be large.
\fi

\begin{lemma}
\label{lem:inner-ham}
Let $S$ be a set of $n$ points for which at most $|S|/7$ points lie on any line, and at most $|S|/7$ points lie on the convex hull.
Then the number of polygonalizations of $S$ is at least singly exponential in $S$.
\end{lemma}

\iffull
\begin{proof}
At least $6|S|/7$ points of $S$ are interior to the convex hull, within which the number of collinear points meets the conditions of  \cref{lem:distinct-paths}. By \cref{lem:distinct-paths}, there are exponentially many distinct non-crossing Hamiltonian paths through the subset of interior points, starting and ending at a vertex of the convex hull of the interior points. It remains to argue that each such path $P$ has a polygonalization $C_P$ that uses the points of $P$ in path order. Because each two polygonalizations constructed in this way have a different ordering on the interior points of $S$, they will be distinct polygonalizations.

\begin{figure}[t]
\centering\includegraphics[scale=0.4]{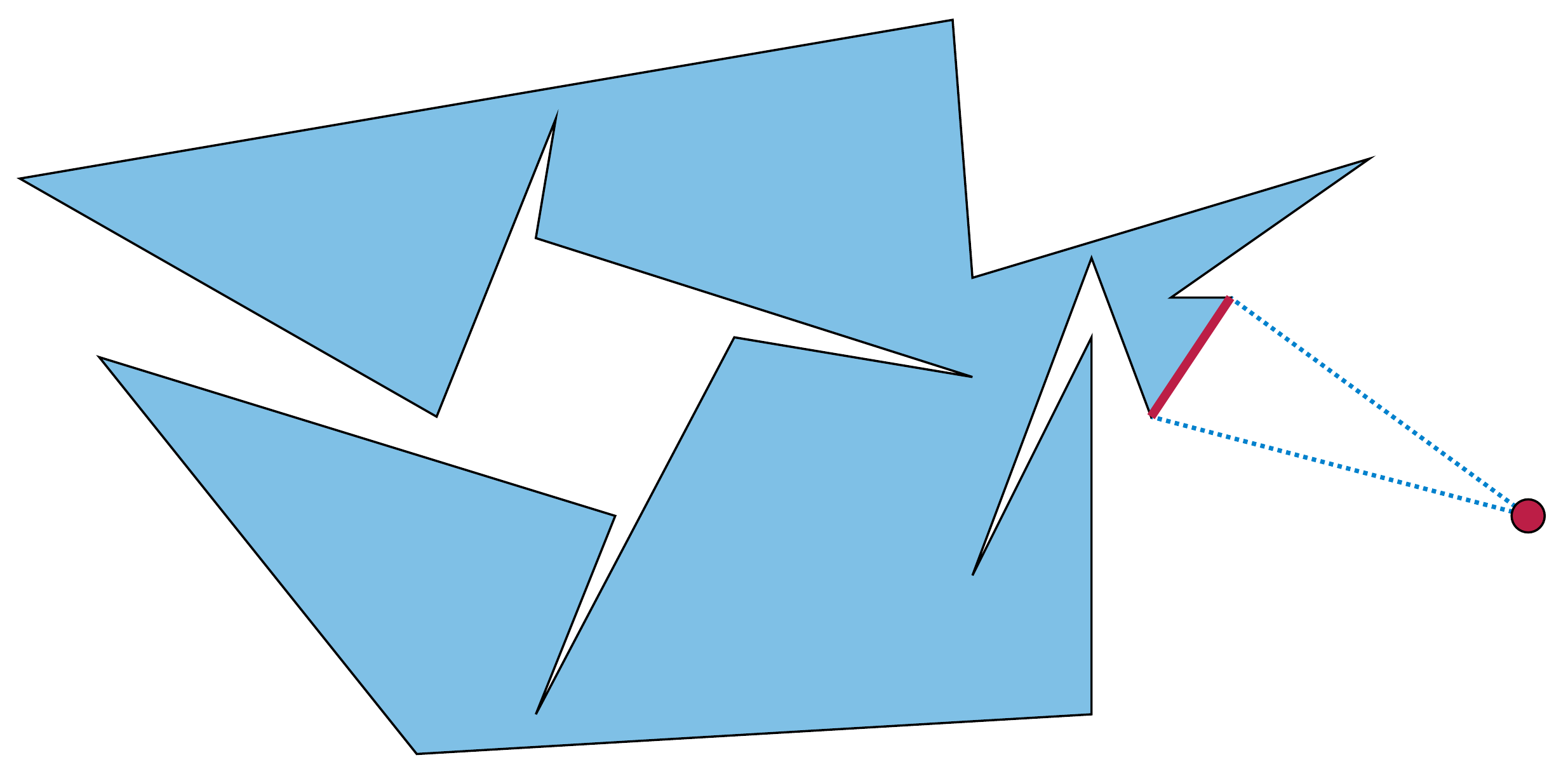}
\caption{Steinhaus completion of a polygonalization: For every polygon (blue) that does not use a given convex hull vertex $q$ (red), some polygon edge will be completely visible to $q$, and the polygon can be modified to include $q$ by replacing that edge by a pair of edges through $q$~\cite{Ste-64}.}
\label{fig:steinhaus}
\end{figure}

Therefore, let $P$ be any non-crossing Hamiltonian path through the subset of interior points, starting and ending at a vertex of the convex hull of the interior points. Triangulate the region between the convex hull of the interior points and the convex hull of $S$, and extend each end of $P$ by an edge from the triangulation; these edges must connect to convex hull vertices of $S$, and cannot cross each other. If they do not meet, complete the extended path to a polygon by following edges around the convex hull of $S$ from one endpoint of the extended path to the other. Let $Q$ be the resulting simple polygon. In most cases, $Q$ will not be a surrounding polygon of $S$, but the only points missing from $Q$ will be some of the convex hull vertices of $S$.

We now apply an argument of Steinhaus~\cite{Ste-64} to show that $Q$ can be completed to a polygonalization $C_P$, preserving the cyclic ordering of the points already in $Q$. Consider each point $q$ of the convex hull of $S$ that is not already in $Q$, in an arbitrary order. Because $q$ is on the convex hull, it cannot be surrounded by a cycle of edges, each partially obscuring the next in the cycle: the visibility ordering of edges of $Q$, as viewed from $q$, is acyclic. Therefore, at least one edge of $Q$ is completely visible to $q$. Replace this edge by a pair of edges through $q$ (gluing a triangle formed by this edge and $q$ onto the polygon; see \cref{fig:steinhaus}). After gluing in all remaining convex hull vertices in this way, the result is a polygonalization of~$S$.
\end{proof}
\fi

\iffull
\subsection{Lower bound for near-convex sets}

For point sets that have few points interior to their convex hull, we have a lower bound on the number of polygonalizations analogous to the bound of \cref{lem:1side} for Hamiltonian paths on near-collinear sets. The main idea of the proof is the same: partition the interior points into a sequence of convex regions, associated with non-adjacent edges of the hull (in place of edges along a line), and show that each such partition can be represented by a polygonalization. However, compared to \cref{lem:1side} some additional care is needed to ensure that the convex region associated with one hull edge does not block visibility to other hull edges and their regions. We do this by splitting the input by a line, with the convex regions on one side of the line and the hull edges they are associated to on the other side. In this way, the visibilities between hull edges and convex regions cannot be blocked. This part of our lower bound does not require a general position assumption.
\fi

\begin{lemma}
\label{lem:poly-few}
Let $S$ be a set of points with $|S|=n$, such that $h$ points of $S$ lie on its convex hull (either as vertices or within its edges). Then
\[
\nham(S)\ge\binom{\lfloor h/4\rfloor + \lceil (n-h)/2\rceil - 1}{\lceil (n-h)/2\rceil}
\]
\end{lemma}

\iffull
\begin{proof}
The points on the convex hull of $S$ partition it into $h$ segments. Consider any diagonal edge $e$ that partitions these segments into two subsequences of equal or nearly equal length; choose arbitrarily one of the two endpoints of $e$ to be first and the other to be second. At least one of the two sides of $e$ will contain at least $\lceil (n-h)/2\rceil$ points that are not on the hull (shown on the right side of $e$ in \cref{fig:hull-slices}); let $I$ denote this set of points. On the other side of $e$, we can find $\lfloor h/4\rfloor$ disjoint segments of the hull; let $D=d_1,d_2,\dots$ denote this sequence of segments, in order from the first endpoint of $e$ to the second endpoint of $e$ (shown in red on the left side of \cref{fig:hull-slices}).

\begin{figure}[t]
\centering\includegraphics[scale=0.3]{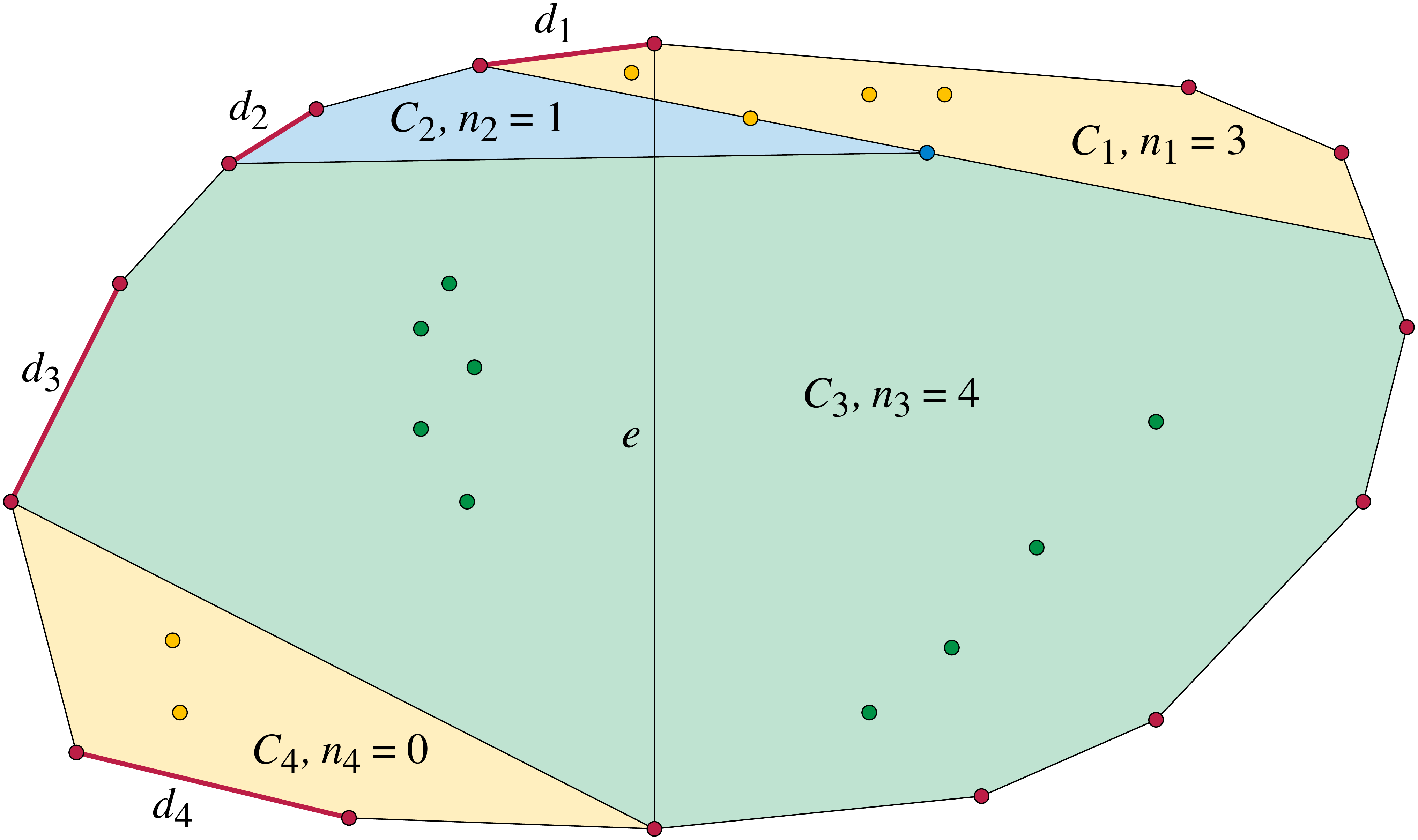}
\caption{Partition of the convex hull of $S$ into convex regions $C_i$ associated with segments $d_i$, for the proof of \cref{lem:poly-few}}
\label{fig:hull-slices}
\end{figure}

We will partition the interior of the convex hull of $S$ into convex regions $C_i$, associated with and including each segment $d_i$ (the four colored regions in \cref{fig:hull-slices}, with interior points colored by which region they belong to). Given such a partition, we can find a polygonalization that modifies the convex hull (a surrounding polygon) by replacing each segment $d_i$ by a Hamiltonian path through the non-hull points in $C_i$, according to \cref{lem:2hull}. Let $n_i=|I\cap C_i|$; then the polygon constructed in this way has $n_i$ points of $I$ between the endpoints of $d_i$, allowing it to be distinguished from any polygon coming from a partition with different values of $n_i$. The numbers $n_i$ are non-negative and sum to $|I|$, and the number of ways of choosing non-negative numbers $n_i$ that sum to $|I|$ is at least equal to the binomial coefficient in the lemma (possibly larger if $D$ or $I$ are larger than their minimum sizes). It remains to prove that each such choice can be represented by a convex partition, and therefore also by a polygonalization.

As in \cref{lem:1side} we construct these convex sets $C_i$ by a greedy process in sequence order. At each step, we maintain the invariant that the remaining points not assigned to a convex set lie within a convex set $K_{i-1}$ formed by the intersection of the convex hull of the input and a half-plane, and that $K_{i-1}$ includes at least one point of $e$. The convex sets $K_i$ are not shown in the figure, but can be reconstructed as the union of the sets $C_j$ for $j>i$. As a base case, let $K_0$ be the convex hull of all of the given points. To construct $C_i$, we consider the following cases.
\begin{itemize}
\item If $d_i$ is the last segment in sequence $D$, let $C_i=K_{i-1}$. In the figure, $C_4$ is constructed by this case.
\item If $d_i$ is not the last segment, but all segments $d_j$ for $j>i$ have $n_j=0$, let $p_i$ be the second endpoint of $d_i$.
Draw line $L_i$ through $p_i$ and the second endpoint of $e$, and cut $K_{i-1}$ along $L_i$ into two convex subsets. Let $C_i$ be the subset containing $d_i$ and let $K_i$ be the other subset. Assign to $C_i$ any points on line $L_i$.
\item If $n_i=0$, draw $L_i$ from the same point $p_i$ to the point on $e\cap K_{i-1}$ closest to the first endpoint of $e$. As in the previous case, and cut  $K_{i-1}$ along $L_i$ into $C_i$ and $K_i$, assigning points on $L_i$ to $C_i$. In the figure, $C_3$ is constructed by this case.
\item In the remaining case, number the points in $I\cap K_{i-1}$ in radial order around $q$, breaking ties in favor of closer points to $q$, and draw line $L_i$ through $q$ and the point in position $n_i$ in this numbered sequence. Cut $C_{i-1}$ along $L_i$ into two convex subsets $C_i$ and $K_i$, where $C_i$ contains $d_i$ and $K_i$. For points that lie on $L_i$, assign the one numbered $n_i$ and any closer points to $C_i$, and assign any farther points to $K_i$. For instance, in the figure, line $L_1$ contains two interior points; the closer yellow point is the one labeled $n_1=3$, and is assigned to $C_1$, while the farther blue point is assigned to $K_1$ and later to $C_2$.
\end{itemize}
In this way, the construction assigns exactly $n_i$ points of $I$ to the convex region $C_i$. These regions partition the convex hull of $S$, in such a way that all interior points of the hull (regardless of whether they belong to $I$) belong to exactly one of these regions. Each assignment can be used to construct a distinct polygonalization, and the number of assignments is at least the value given in the lemma.
\end{proof}
\fi

\subsection{Putting the bounds together}
Combining our lower bounds into a single formula, we have:

\begin{lemma}
\label{lem:cycle-lower}
Let $S$ be a set of $n$ points with $\min(\offline(S),\inhull(S))=m$. Then
\[ \log\npoly(S)=\Omega\left(m\left(\log\frac{n}{m}+1\right)\right).\]
\end{lemma}

\begin{proof}
We consider the following cases:
\begin{itemize}
\item If $\inhull(S)\le 6n/7$, the result follows from  \cref{lem:poly-few}. In particular this applies when the largest subset of collinear points in $S$ has size $\ge n/7$ and is part of the convex hull.
\item If the largest subset of collinear points in $S$ has size $\ge n/7$ but is not part of the convex hull, then let $L$ be the line through this subset. Because $L$ does not lie on the convex hull, $S\setminus L$ includes points on both sides of $L$, with at least $m/2$ points in one of these two halfplanes. By \cref{lem:1side} the number of Hamiltonian paths through the points in $L$ and the points in this halfplane, starting and ending at the two extreme points of $L$, meets or exceeds the lower bound in the statement of this lemma. Each of these Hamiltonian paths can be completed to a polygonalization through the points in the other halfplane bounded by $L$, by \cref{lem:2hull}.
\item In the remaining case, the largest subset of collinear points in $S$ has size $< n/7$ and $\inhull(S)\ge 6n/7$. In this case, $m=\Omega(n)$ and the bound of the lemma reduces to $\Omega(n)$. The result follows from  \cref{lem:inner-ham}.
\end{itemize}
Since all cases have at least the number of polygonalizations stated, the bound holds.

For a bound of $\Omega(i)$, apply \cref{lem:inner-ham}, and for $\Omega(i\log n/i)$ when $i\le n/2$, apply \cref{lem:poly-few}, in both cases using \cref{lem:log-binom} to estimate the logarithm of the binomial coefficient.
\end{proof}

From the fact that the upper bound of \cref{cor:cycle-upper} and the lower bound of \cref{lem:cycle-lower}  have exactly the same form, we obtain a constant factor approximation to the logarithm of the number of surrounding cycles and of polygonalizations. For our bound on the complexity of the algorithm for listing polygonalizations we need it in the following form:

\begin{theorem}
\label{thm:cycle-equivalence}
For a given point set $S$,
\[ \nsurround(S)=\npoly(S)^{O(1)}.\]
\end{theorem}

\section{Nonlinearity}

In this section we investigate the exponent of the polynomial bounds on non-crossing paths as a function of non-crossing Hamiltonian paths, and on surrounding cycles as a function of polygonalizations. In both cases we show that the exponent is bounded away from one, for inputs in which the number of non-crossing configurations is exponential. This implies, in particular, that the backtracking algorithms that we investigate for Hamiltonian paths and polygonalizations can in some instances be forced to take an amount of time per output that is exponential in the input size, despite being polynomial in the output size.

The construction for paths and Hamiltonian paths is very simple:

\begin{theorem}
There exist sets $S$ of points for which
\[\npath(S)\ge \nham(S)^{\log_2 3-o(1)}\approx\nham(S)^{1.585}\]
and moreover for which both the number of non-crossing paths and the number of non-crossing Hamiltonian paths is exponential in $|S|$.
\end{theorem}

\begin{proof}
We may take $S$ to be in convex position. If a set $S$ in convex position has $n$ points, it has $n2^{n-3}$ non-crossing Hamiltonian paths~\cite{A001792}, but
\[\frac{n}{4}(3^{n-1}+3)\]
non-crossing paths in total, considering a single vertex to count as a path of length zero.

\begin{figure}[t]
\centering\includegraphics[width=0.45\textwidth]{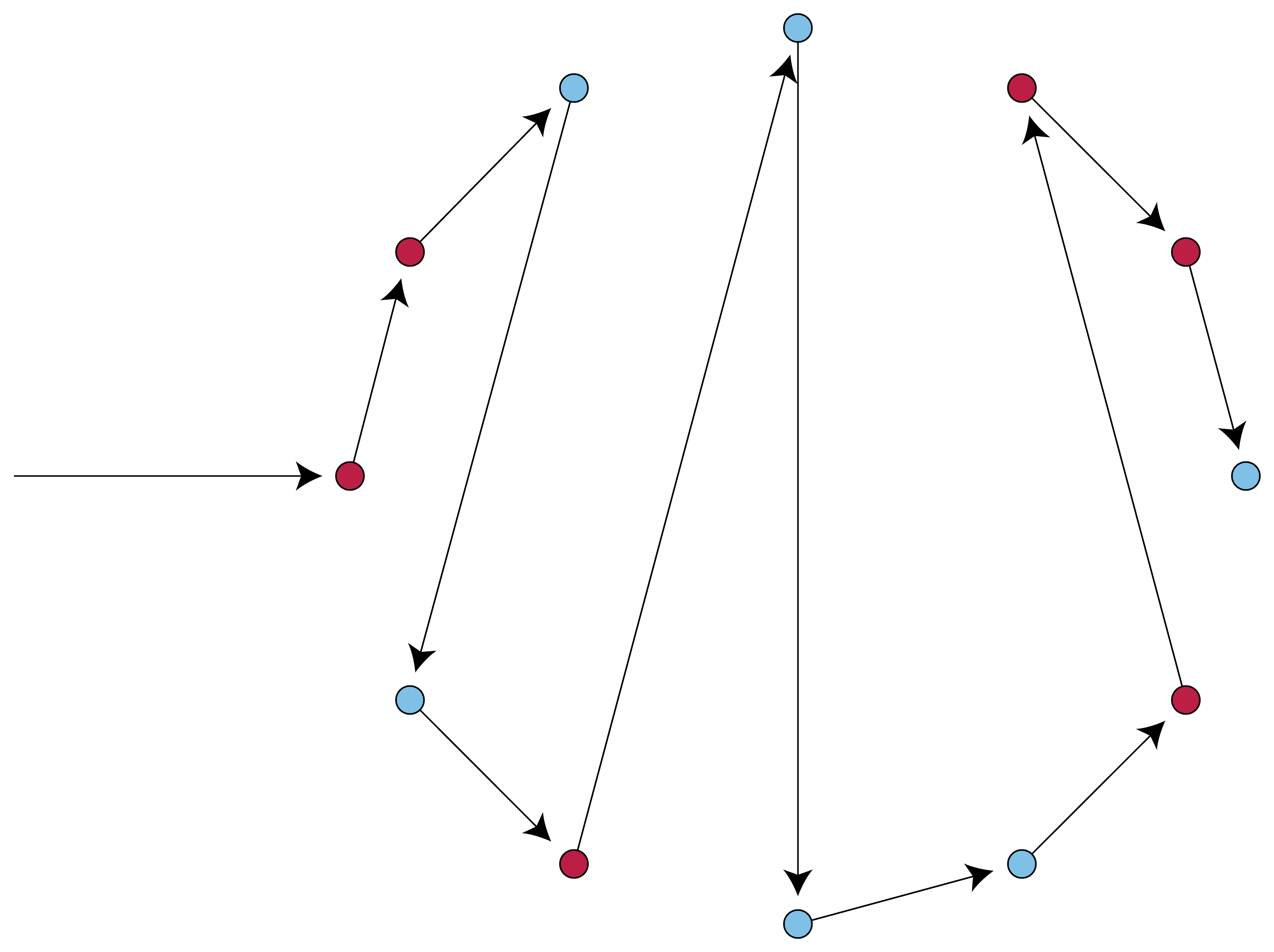}
\qquad\includegraphics[width=0.45\textwidth]{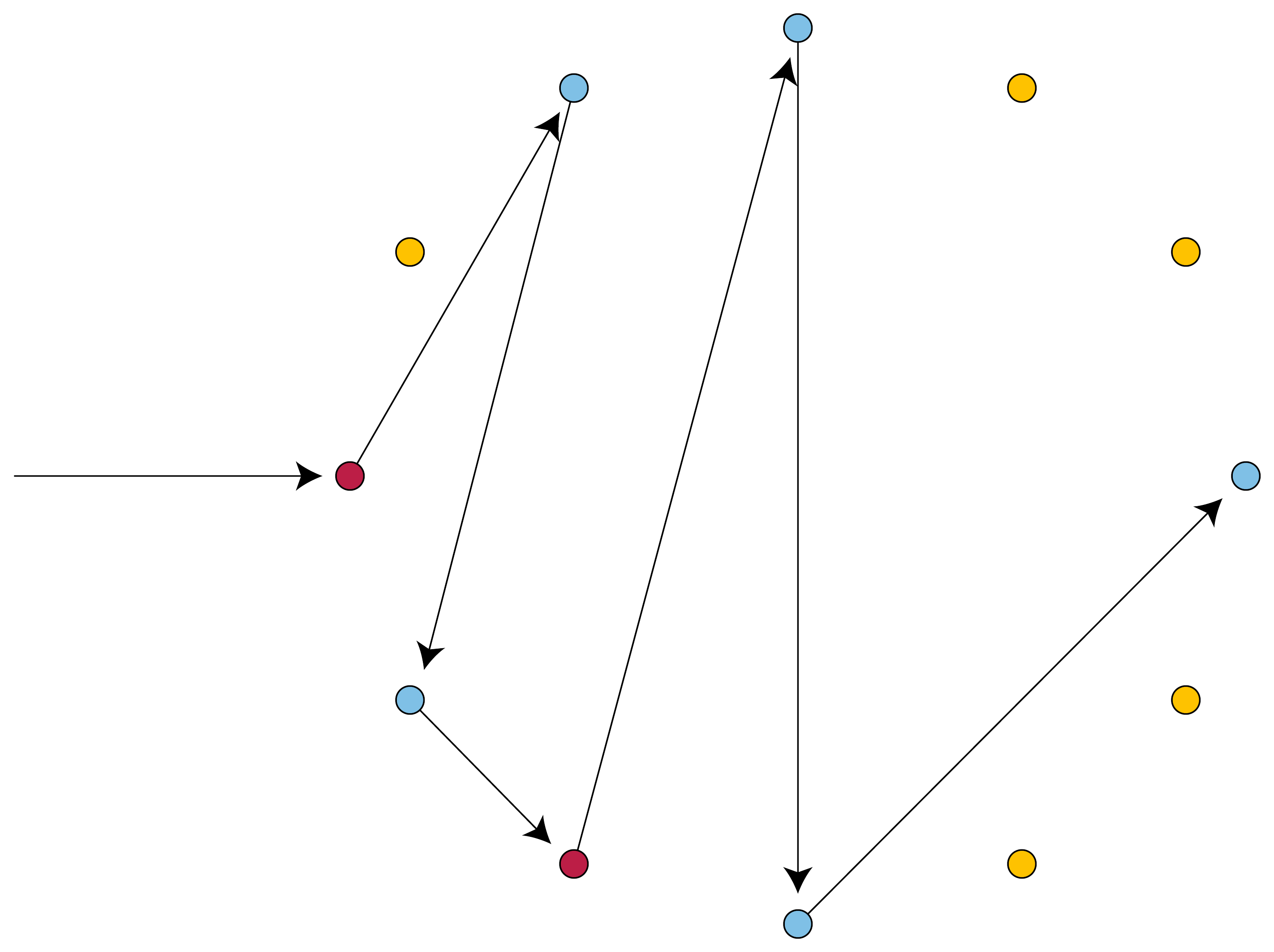}
\caption{Coloring-based arguments for the number of non-crossing paths and Hamiltonian paths of a point set in convex position. Each 2-coloring of the points and choice of starting point determines a non-crossing Hamiltonian path in which the colors determine the direction of the next step; each 3-coloring and choice of starting point determines a path, skipping vertices of one of the colors.}
\label{fig:colored}
\end{figure}

Both bounds can be proven by a simple coloring argument (\cref{fig:colored}). For non-crossing Hamiltonian paths, consider the $2^n$ ways of coloring the points red and blue, and $n$ choices of where to start. For each choice, follow a path that, at each red point,  steps clockwise to the next available vertex, and at each blue point steps counterclockwise. Each non-crossing Hamiltonian path is found for exactly eight choices: the path can start at either end, and the final two vertex colors are irrelevant. Thus, the number of paths is $n2^n/8$.

For the bound on non-crossing paths, consider instead the $3^n$ ways of coloring the points red, blue, and yellow, and skip the yellow points both in choosing a starting point and in considering which vertices are available in subsequent steps of the path. Cyclically permuting the colors in a coloring groups the $3^n$ colorings into $3^{n-1}$ orbits, each of which has a total of $2n$ red or blue starting points, so there are $2n\cdot 3^{n-1}$ choices. Again, each path is found for exactly eight choices, except for the single-vertex paths, which are found for only two choices. The total number of paths is obtained by dividing $2n\cdot 3^{n-1}$ by eight, and adjusting for the number of single-vertex paths.
\end{proof}

For an analogous separation between surrounding polygons and polygonalizations, we replace one edge of a triangle by a convex chain of $n-2$ edges, within the triangle (\cref{fig:pseudotriangle}). The polygonalizations of this point set are obtained from non-crossing Hamiltonian paths through the points of the convex chain, with both ends of the path connected to the one point that does not belong to the convex chain (which we call the \emph{apex}). The edges to the apex cannot cross any other edge, so all Hamiltonian paths lead to polygonalizations in this way. Therefore, by the same formula already used above, the number of polygonalizations is exactly $(n-1)2^{n-4}$.

A surrounding polygon for these points must have the apex as a vertex (because it lies on the convex hull). It
 must also include as vertices all of the points of the convex chain that lie outside the two neighbors of the apex. The points of the convex chain that lie between the two neighbors of the apex may either be vertices of the surrounding polygon, or omitted; if omitted, they will automatically be surrounded. We may parameterize a surrounding polygon by three non-negative numbers: $a$, the number of \emph{outer points} of the convex chain that lie outside the two neighbors of the apex, $b$, the number of \emph{inner points} that lie between these two neighbors and are vertices of the polygon, and $c$, the number of \emph{omitted points} that are not vertices of the surrounding polygon. Necessarily, $a+b+c=n-3$, because the points that are counted by these three numbers include all points except the apex and its two neighbors.
 
\begin{figure}[t]
\centering\includegraphics[width=0.8\textwidth]{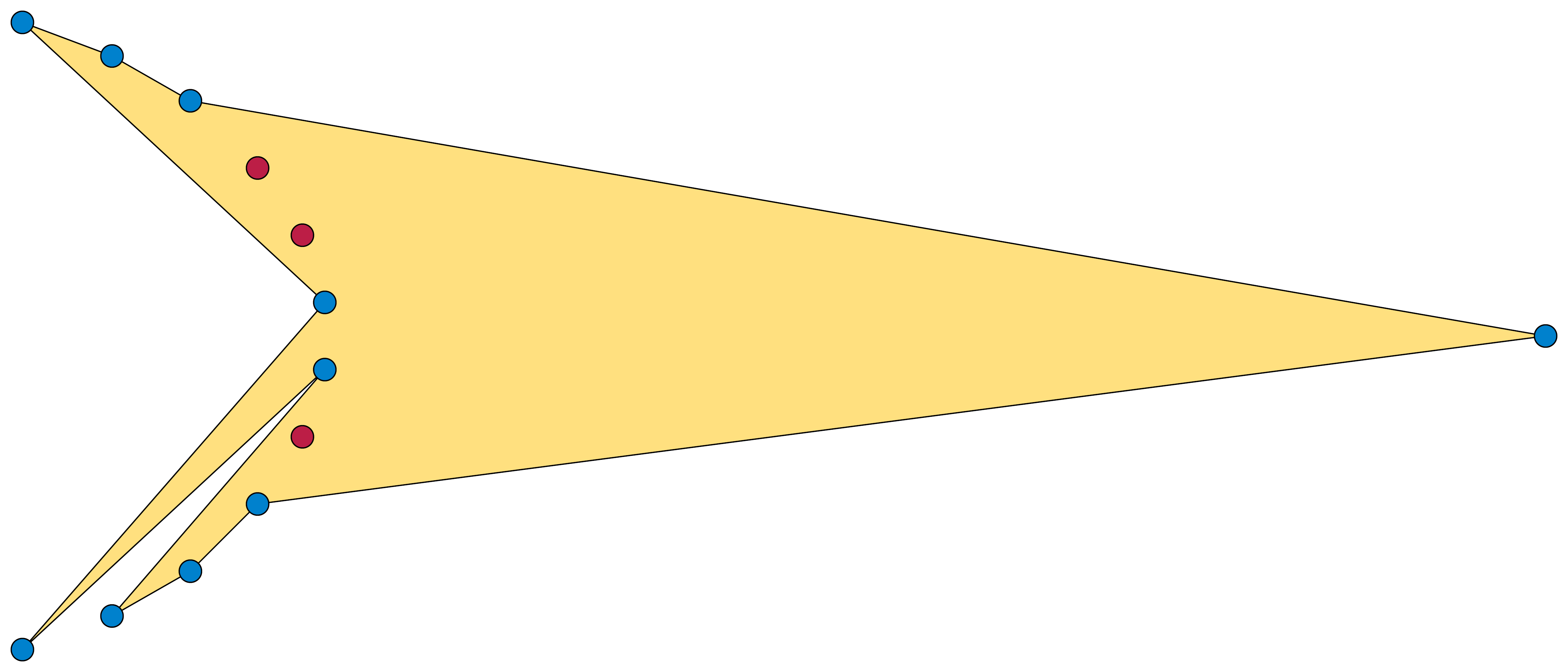}
\caption{Points formed by replacing an edge of a triangle by a convex chain of points within the triangle, and a surrounding polygon (yellow) for these points. Points can be omitted as polygon vertices (red) only when they lie between the two neighbors of the apex (the rightmost point of the figure).}
\label{fig:pseudotriangle}
\end{figure}

Once $a$, $b$, and $c$ have been chosen, the surrounding polygon itself may be determined by three more choices: how to partition the outer $a$ points into left and right subsets, with $a+1$ possibilities, how to alternate between outer and inner points along the polygon, with $\tbinom{a+b}{a}$ possibilities, and how to partition the points between the two neighbors of the apex into inner points and omitted points, with $\tbinom{b+c}{b}$ possibilities. Therefore, the total number of surrounding polygons of this point set is
\[
\sum_{a+b+c=n-3} (a+1)\binom{a+b}{a}\binom{b+c}{b}.
\]
These numbers, for $n=3, 4, 5, \dots$, are
\[
1, 4, 13, 40, 120, 354, 1031, 2972, 8495, 24110, \dots,
\]
a sequence having the generating function $(1-2x)/(1-3x+x^2)^2$ and growing proportionally to $n(\varphi+1)^n$, where $\varphi=(1+\sqrt5)/2$ is the golden ratio~\cite{A238846}. This example proves:

\begin{theorem}
There exist sets $S$ of points for which
\[\nsurround(S)\ge \npoly(S)^{\log_2 (\varphi+1)-o(1)}\approx\npoly(S)^{1.388}\]
and moreover for which both the number of surrounding polygons and the number of polygonalizations is exponential in $|S|$.
\end{theorem}

\section{Conclusions}

We have developed simple output-sensitive algorithms for listing all non-crossing Hamiltonian paths and all polygonalizations for a point set. However, their dependence on the output size is polynomial, not linear. It would be of interest to find alternative algorithms with a better dependence on the output size, as well as more accurate approximations for the numbers of non-crossing structures.

\bibliographystyle{plainurl}
\bibliography{polygonalization}
\end{document}